\documentclass[11pt,notitlepage,tightenlines,nofootinbib,superscriptaddress,aps,pra]{revtex4-2}
\usepackage[cal=boondox]{mathalfa}        
\usepackage[margin=1in]{geometry}
\usepackage{amsmath,amssymb,amsthm,microtype,mathtools}
\usepackage{enumitem}
\usepackage{graphicx}
\usepackage{float}
\usepackage[section]{placeins}
\usepackage{xcolor}
\usepackage{tikz}
\usepackage{pgfplots}
\usetikzlibrary{arrows.meta,calc,decorations.pathreplacing,positioning,fit,backgrounds,shapes.geometric,patterns,shadows,shapes.misc}
\usepackage{comment}
\usepackage{mathpazo}
\usepackage{braket}
\usepackage{booktabs}

\definecolor{mitred}{RGB}{163, 31, 52} 
\usepackage[colorlinks=true,linkcolor=mitred,citecolor=mitred,urlcolor=mitred]{hyperref}

\pgfplotsset{compat=1.18}
\pgfmathdeclarefunction{exactpt}{1}{%
  \pgfmathparse{0.5*(1 + floor(1/(#1))*(#1)^2 + (1-floor(1/(#1))*(#1))^2)}%
}

\numberwithin{equation}{section}

\newtheorem{theorem}{Theorem}[section]
\newtheorem{lemma}[theorem]{Lemma}
\newtheorem{claim}[theorem]{Claim}
\newtheorem{proposition}[theorem]{Proposition}
\newtheorem{corollary}[theorem]{Corollary}
\newtheorem{definition}[theorem]{Definition}

\DeclareMathOperator{\Id}{\mathbb{I}}
\DeclareMathOperator{\SYT}{SYT}

\newcommand{\C}{\mathbb C}
\newcommand{\abs}[1]{\left\lvert #1\right\rvert}
\newcommand{\norm}[1]{\left\lVert #1\right\rVert}
\newcommand{\tr}[1]{\textnormal{tr}\left[#1\right]}
\newcommand{\ketbra}[2]{\ket{#1}\!\!\bra{#2}}

\newcommand{\swap}{\mathbb{F}}

\newcommand{\eps}{\varepsilon}

\newcommand{\BC}{\mathcal{B}}

\newcommand{\EC}{\mathcal{E}}

\newcommand{\HC}{\mathcal{H}}
\newcommand{\cH}{\mathcal{H}}
\newcommand{\IC}{\mathcal{I}}

\newcommand{\LC}{\mathcal{L}}
\newcommand{\cL}{\mathcal{L}}

\newcommand{\PC}{\mathcal{P}}
\newcommand{\cP}{\mathcal{P}}
\newcommand{\QC}{\mathcal{Q}}
\newcommand{\cQ}{\mathcal{Q}}

\newcommand{\SC}{\mathcal{S}}
\newcommand{\TC}{\mathcal{T}}
\newcommand{\UC}{\mathcal{U}}
\newcommand{\cU}{\mathcal{U}}

\DeclareMathOperator{\End}{End}
\DeclareMathOperator{\trc}{tr}

\begin{document}

\title{Tight Universal Bounds on Quantum Data Hiding with Multipartite Werner States}
\author{Oren Akresh}
\affiliation{Central High School, Champaign, IL 61820, USA.}

\author{Jacob Beckey}
\email{jbeckey@illinois.edu}
\affiliation{Department of Mathematics, University of Illinois at Urbana-Champaign, Urbana, IL 61801, USA.}
\affiliation{Department of Physics, University of Rhode Island, Kingston, RI 02881, USA}

\author{Felix Leditzky}
\affiliation{Department of Mathematics, University of Illinois at Urbana-Champaign, Urbana, IL 61801, USA.}
\date{\today}

\begin{abstract}
Quantum data hiding concerns pairs of states which are highly distinguishable with global measurements, yet nearly indistinguishable when restricted to local operations and classical communication. More than two decades after Eggeling and Werner introduced a multipartite hiding scheme based on Werner states, the optimal dependence of its uniform security guarantee on the number of parties and local dimension remained unknown. We resolve this problem by showing that the distinguishing bias of any pair of $n$-qudit Werner states is $O(n^2/d)$ under measurements whose effects remain positive under partial transposition (PPT) across every bipartition. An explicit pair attains this scaling using only nonadaptive local measurements, establishing worst-case optimality and showing that the PPT relaxation preserves the optimal dependence on both the number of parties and the local dimension. At fixed security, this extends the certified hiding regime from $n=O(d^{1/4})$ to $n=O(\sqrt d)$. Beyond data hiding, the same bound implies that testing any nontrivial unitarily invariant property with adaptive single-copy measurements requires $\Omega(\sqrt d)$ copies, yielding separations for any property with dimension-independent sample complexity under collective measurements. Our proof reduces the distinguishing bias to trace norms of partially transposed operators and analyzes them using mixed Schur-Weyl duality, demonstrating the utility of representation theoretic methods developed for port-based teleportation to data hiding and quantum property testing.
\end{abstract}

\maketitle

\newpage
\tableofcontents
\clearpage

\section{Introduction}
Quantum state discrimination is one of the most fundamental tasks in quantum information processing~\cite{helstrom1969Quantum}. For the binary discrimination problem, where one is given either $\rho_0$ or $\rho_1$ with equal probability, Holevo-Helstrom tells us that the optimal distinguishing probability is 
\begin{align}
    p_{\rm succ} \leq \frac{1}{2} + \frac{1}{4}\|\rho_0-\rho_1\|_1,
\end{align}
giving the trace distance its operational meaning and proving that, if $\|\rho_0-\rho_1\|_1$ is large, there exists a \textit{global} measurement which can distinguish the two states. However, it has been known since the early days of quantum information that one's ability to discriminate two quantum states is highly dependent upon what type of measurements are permitted~\cite{peres1991Optimal}. The extreme version of this fact is that there exist states which are perfectly distinguishable by global measurements, but nearly indistinguishable by local measurements~\cite{bennett1999Quantum}. Around the turn of the century, Terhal, DiVincenzo, and Leung leveraged this fundamental observation to construct a cryptographic primitive now referred to as \textit{quantum data hiding}~\cite{terhal2001Hiding,divincenzo2002Quantum}. In quantum data hiding, one encodes a classical bit in the choice of one of two density operators such that they are readily distinguishable by global measurements while being highly indistinguishable using only local measurements (i.e. the classical data is ``hidden" unless one has access to the entire state). Since the original paper on the topic, many different quantum data hiding schemes have been devised~\cite{eggeling2002Hiding,hayden2004Randomizing,hayden2005Multiparty,lancien2013Distinguishing,lami2021Quantum,wang2025Gaussian,mele2025Optimising}. 

In this work, we focus on one of the longest standing data hiding schemes, due to Eggeling and Werner, in which one hides a classical bit in the choice of two multipartite mixed states satisfying collective unitary symmetry~\cite{eggeling2002Hiding,eggeling2003multipartite}. These so-called multipartite Werner states have a particularly simple form afforded by Schur-Weyl duality, making them amenable to rigorous security analysis~\cite{werner1989Quantum,eggeling2002Hiding}. Despite this fact, the optimal dependence of its uniform security guarantee on the number of parties and local dimension was not known, with the best known bound being due to Harrow and being proved using approximate orthogonality of permutation operators~\cite{harrow2023Approximate}. Before sketching brief history of data hiding with Werner states, let us state another fundamental motivation for this work.

Another ubiquitous application for quantum state discrimination arises in quantum learning theory, where one is concerned with determining the resources that are necessary and sufficient to characterize some property of an unknown quantum system. Particularly relevant to the present work is the field of quantum property testing, in which one is given an unknown quantum state and asked to determine whether it satisfies some property (i.e. purity, productness, etc.) or is far from any such state, given the promise that one of these is the case~\cite{montanaro2016Survey}. A canonical method for proving sample complexity lower bounds for property testers reduces binary quantum state discrimination to the property testing task at hand. The idea is quite simple: construct two states which, if a tester existed, could be distinguished by the tester. Then, a lower bound on the sample complexity of the discrimination task immediately implies a lower bound on the sample complexity of the testing task.

This connection provides very strong motivation for studying distinguishability with restricted measurements. In recent years, such reductions have been used to show that, for many learning and testing tasks, collective measurements require a number of copies independent of the dimension, while even adaptive measurements of one copy at a time have dimension-dependent sample complexity~\cite{chen2022Exponential,huang2022Quantum,liu2024Quantum,ye2025Exponential,noller2025infinite,beckey2025Product}. Testing lower bounds can also imply lower bounds for estimation tasks: an estimate accurate enough to separate states with and without the property would itself give a tester. Thus, restricted-measurement indistinguishability can therefore establish sample complexity separations for both property testing and estimation. These separations are theoretically interesting because they reveal how much collective measurements can outperform measurements on individual copies, but they also have practical relevance because testing and estimation underpin the characterization and benchmarking of quantum devices, which are essential to the development of quantum technologies.

Proving these sample complexity lower bounds against adaptive single-copy protocols is non-trivial, in large part, because of the role of adaptivity. To overcome the challenge of handling adaptive protocols, Refs.~\cite{aharonov2022Quantum, chen2022Exponential} developed a framework, inspired by standard tools in theoretical computer science, which allows one to derive sample complexity lower bounds in many situations. This so-called \textit{learning tree formalism} can sometimes lead to extremely technical analysis of classical probability distributions, requiring tools less familiar to a broad physics audience and, in some cases, being intractable~\cite{hinsche2025SingleCopy}. As we describe below, our methods complement the learning tree formalism, utilize tools more ubiquitously known throughout quantum information theory, and in several cases, provide stronger results with remarkably simple proofs. 

With these motivations in mind, we briefly review the history of data hiding using Werner states before informally stating our main results and sketching our proof technique.

\subsection{A Brief History of Quantum Data Hiding with Werner States}
To understand the working principle of Terhal, DiVincenzo, and Leung's original data hiding scheme, consider a hider, Charlie, trying to hide one classical bit from Alice and Bob by encoding the bit in a Bell state, e.g.
\begin{align}
    0 \mapsto \ket{\Phi^+}=\frac{\ket{00}+\ket{11}}{\sqrt{2}}, \quad 1 \mapsto \ket{\Psi^-} = \frac{\ket{01}-\ket{10}}{\sqrt{2}}.
\end{align}
Charlie gives one qubit of the Bell pair to Alice and the other to Bob. At first glance, the information seems secure: each party's reduced state is $\Id/2$, so neither party can obtain any information about the hidden bit from their local system alone. Once classical communication is allowed, however, Alice and Bob can distinguish the states perfectly. They simply measure their respective qubits in the computational basis and compare their outcomes over the phone: the outcomes agree for $\ket{\Phi^+}$ and disagree for $\ket{\Psi^-}$. This example illustrates a general result that precludes using two \textit{orthogonal} pure states for data hiding: any two orthogonal pure states can be perfectly distinguished by local operations and classical communication~\cite{walgate2000Local}.

This simple example demonstrates the need to encode the classical data into mixed states if one demands both perfect global recoverability and security against arbitrary LOCC protocols. To this end, Terhal, DiVincenzo, and Leung devise the following scheme~\cite{terhal2001Hiding,divincenzo2002Quantum}. Charlie uniformly samples $m$ Bell states, conditioned on the total number of singlets being even when $b=0$ and odd when $b=1$, and distributes one qubit from each pair to Alice and the other to Bob. Each triplet is a $+1$ eigenstate of the two-qubit swap, whereas the singlet is a $-1$ eigenstate. Consequently, the parity of the number of singlets determines whether the joint state lies in the symmetric or antisymmetric subspace under the global swap $\mathbb{F}_{AB}$ between Alice's and Bob's $m$-qubit registers. Letting $D \coloneq 2^m$, the resulting data-hiding states may be expressed as
\begin{align}
    \rho^{(m)}_0 = \frac{\Pi_{\rm sym}}{\tr{\Pi_{\rm sym}}}= \frac{\Id + \mathbb{F}_{AB}}{D(D+1)} \quad \text{and} \quad \rho^{(m)}_1 = \frac{\Pi_{\rm asym}}{\tr{\Pi_{\rm asym}}}= \frac{\Id - \mathbb{F}_{AB}}{D(D-1)},
\end{align}
which are the, now well-known, normalized symmetric and antisymmetric Werner states\footnote{For a pedagogical introduction to Werner states, including both the
two-qubit example and the general qudit construction, see
Ref.~\cite[Secs.~3.1 and 6.1]{leditzky2026Representation}.} on $\mathbb{C}^D \otimes \mathbb{C}^D$~\cite{werner1989Quantum}. Their supports are orthogonal, so the global projective measurement $\{\Pi_{\rm sym},\Pi_{\rm asym}\}$ distinguishes them perfectly. However, the authors of Ref.~\cite{terhal2001Hiding} show that any two-outcome POVM $\{M,\Id-M\}$ implementable by LOCC between Alice and Bob will yield a distinguishing bias satisfying
\begin{align}
   \abs{\tr{M(\rho_0^{(m)}-\rho_1^{(m)})}} \leq \frac{2}{D} = \frac{1}{2^{m-1}}.
\end{align}
For equal prior probabilities, Alice and Bob's optimal success probability is therefore bounded by $ p_{\rm succ}^{\rm LOCC} \leq \frac{1}{2}+\frac{1}{D}.$ Thus, their advantage over random guessing vanishes exponentially in the number of Bell pairs used to construct the hiding states.

Following this work, Eggeling and Werner generalized quantum data hiding to multipartite Werner states on $(\mathbb{C}^d)^{\otimes n}$~\cite{eggeling2002Hiding,eggeling2003multipartite}. In this scheme, they encode a classical bit $b \in \{0,1\}$ into one of two density operators $\rho_0$ and $\rho_1$ satisfying collective-unitary symmetry
\begin{align}
    [\rho_b,U^{\otimes n}] = 0 \quad \text{for all} \quad U \in \UC_d.
\end{align}
Eggeling and Werner showed that, for a fixed number of parties $n$, multipartite Werner states can be chosen to be nearly perfectly indistinguishable by global measurements but have distinguishing bias $O_n(1/d)$ under the measurement effects which remain positive under partial transposition (PPT) of any subset of indices $S \subseteq [n]$. This PPT relaxation is  physically motivated because LOCC is strictly contained within the set of PPT operations~\cite{bennett1999Quantum, chitambar2014Everything}. Moreover, unlike working directly with LOCC protocols, analyzing protocols under the PPT relaxation is often mathematically tractable. 

The central idea in their security proof, similar to Ref.~\cite{terhal2001Hiding}, is that, because the data-hiding pair satisfy a particular symmetry, it suffices to consider measurement effects satisfying this symmetry. Indeed, writing $\Delta \coloneq \rho_0 - \rho_1$, one may twirl any admissible effect $M$ without changing the bias
\begin{align}
    \tr{M \Delta} = \tr{\overline{M}\Delta}, \quad \text{where} \quad \overline{M} \coloneq \underset{U \in \UC_d}{\mathbb{E}}[U^{\otimes n}M (U^\dagger)^{\otimes n}].
\end{align}
The twirling operation preserves the PPT constraints, and by the left-invariance of the Haar measure, one can show that $[\overline{M}, U^{\otimes n}]$ for all $U \in \UC_d$. Schur-Weyl Duality (cf. Cor.~4 in Ref.~\cite{harrow2013Church}) then implies that 
\begin{align}\label{eq:M-expansion}
    \overline{M} = \sum_{\pi \in \SC_n} m_{\pi} P_d(\pi),
\end{align}
where $P_d(\pi)$ is the unitary representation of the permutation $\pi \in \SC_n$~\cite{harrow2013Church}. Using triangle inequality and H\"older's inequality, and noting $\tr{\Delta}=0$, one obtains
\begin{align}
    \abs{\tr{M\Delta}} &= \abs{\sum_{\pi \neq e} m_{\pi} \tr{P_d(\pi) \Delta}} \leq \sum_{\pi \neq e} \abs{m_{\pi}} \|P_d(\pi)\|_{\infty} \|\Delta \|_1 \leq 2 \sum_{\pi \neq e} \abs{m_{\pi}}.
\end{align}
Eggeling and Werner combine the pairwise approximate orthogonality of distinct permutation operators along with the PPT constraint, to show that $\abs{m_{\pi}} = O_n(1/d)$ for all $\pi \neq e$. Noting that $\tr{\Delta}=0$, this yields
\begin{align}
    \abs{\tr{M \Delta}} \leq 2 \sum_{\pi \neq e} \abs{m_{\pi}} = O_n(1/d).
\end{align}
Crucially, this bound only holds with $n$ fixed as $d \rightarrow \infty$. Harrow subsequently proved a uniform $n$-dependent upper bound on this bias as one application of his more general, collective notions of the approximate orthogonality of permutation operators~\cite{harrow2023Approximate}. Concretely, he showed that for any two multipartite Werner states $\rho_0, \rho_1$ and any PPT-BOTH measurement $\{M,\Id-M\}$, one has
\begin{align}
    \abs{\tr{M \Delta}} \leq \frac{6 n^2}{\sqrt{d}}.
\end{align}
Thus, any pair of globally distinguishable multipartite Werner states forms a data-hiding pair with vanishing bias whenever $n = o(d^{1/4})$. Moreover, because this holds for arbitrary Werner state pairs, Schur-Weyl duality and a dimension-counting argument imply the existence of $\sqrt{n!}$ mutually orthogonal pairs, implying that one could simultaneously encode $ \frac{1}{2}\log_2{n!}$ classical bits in multipartite Werner states with pairwise PPT indistinguishability guaranteed by this bound. Thus, Harrow's contribution was to replace the fixed-$n$ upper bound of Eggeling and Werner to an explicit bound that remains meaningful when $n$ grows with $d$.

The main result of this paper is a tightening of Harrow's upper bound to $O(n^2/d)$, which recovers the $1/d$ dependence of Eggeling and Werner while also retaining explicit polynomial dependence on $n$. Moreover, we exhibit a particular pair of multiparite Werner states for which a simple LOCC measurement achieves bias $\Omega(n^2/d)$. Thus, the $n^2/d$ scaling is optimal as a uniform bound over multipartite Werner states. As discussed in the previous section, this then has direct implications for the sample complexity of various property testing tasks. We now overview our main results and sketch their proofs at a very high level.

\subsection{Proof Sketch and Overview of Main Results} \label{sec:overview-main-results}
As mentioned above, a key step in each of the aforementioned security proofs is the use of the symmetry of the data-hiding pair to restrict ones attention to measurement effects with that same symmetry. As mentioned above, Harrow's approach proceeds by expanding a symmetrized PPT-BOTH measurement effect as in Eq.~\eqref{eq:M-expansion} and, after triangle inequality and H\"older's, bounding the magnitude of the individual permutation coefficients uniformly yields a bound on the bias of
\begin{align}
   \abs{\tr{M \Delta}} \leq 2 \left(1+\frac{n^2}{\sqrt{d}}\right) \sum_{\pi \neq e} d^{-\abs{\pi}/2}.
\end{align}
where $|\pi|$ is the minimum number of transpositions necessary to obtain $\pi$ from the identity permutation $e$~\cite{harrow2023Approximate}. To see why this bound cannot yield the desired $1/d$ dependence, note that there are $\binom{n}{2}$ transpositions satisfying $|\pi|=1$. Thus, the leading contribution to this sum is $O(n^2/\sqrt{d})$. To obtain the desired scaling via this term-wise bound, one would need $|m_{\pi}|= O( d^{-\abs{\pi}})$. In the appendix, we construct a simple PPT-BOTH measurement with $|m_{(123)}| = \Theta(1/d)$, while $d^{-\abs{(123)}}=d^{-2}$, precluding a universal coefficient bound of the form $|m_{\pi}|= O( d^{-\abs{\pi}})$.

Here, we take a distinctly different approach. Instead of utilizing the symmetry of the states to expand $M$ in the permutation basis, we will instead expand the difference of states in a telescoping sum and then apply the PPT constraint via two elementary operator identities which, taken together, reduce the problem to bounding the trace norm of a partially transposed operators that arises in this decomposition. To understand the working principle, let $\rho_0, \rho_1$ be multipartite Werner states, i.e. density operators on $(\mathbb{C}^d)^{\otimes n}$ satisfying $[\rho_b, U^{\otimes n}]=0$ for all unitaries $U \in \UC_d$. Let $\tau_d \coloneq \Id/d$ and denote the marginal obtained by tracing out qudits $m+1$ to $n$ as $ \rho^{(m)} \coloneq \operatorname{tr}_{m+1,\ldots,n}(\rho).$ To bound the bias, we will first add and subtract the maximally mixed state and use triangle inequality to write
\begin{align}
    \abs{\tr{M(\rho_0-\rho_1)}} & \leq \abs{\tr{M(\rho_0 -\tau_d^{\otimes n})}} + \abs{\tr{M(\rho_1 -\tau_d^{\otimes n})}}.
\end{align}
Our goal is now to bound the bias achievable when distinguishing a multipartite Werner state $\rho$ from the maximally mixed state. To this end, we use a telescoping sum to write
\begin{align}
    \rho- \tau_d^{\otimes n} &= \rho - \rho^{(n-1)}\otimes \tau_d  + \rho^{(n-1)}\otimes \tau_d - \dotsm - \rho^{(2)} \otimes \tau_d^{\otimes (n-2)}+\rho^{(2)} \otimes \tau_d^{\otimes (n-2)} - \tau_d^{\otimes n}.
\end{align}
Noting that $[\rho,U^{\otimes n}]=0 \implies \rho^{(1)} = \tau_d$ and $\rho^{(n)} \coloneq \rho$. Then, we may write this difference as
\begin{align}
    \rho -\tau_d^{\otimes n} = \sum_{m=2}^n (\rho^{(m)} - \rho^{(m-1)}\otimes \tau_d) \otimes \tau_d^{\otimes (n-m)}.
\end{align}
Note that because it is the difference of two quantum states, each summand is Hermitian and traceless. With this in mind, we may insert this expansion into the bias and use triangle inequality again to obtain
\begin{align}
     \abs{\tr{M(\rho_0-\rho_1)}} \leq \sum_{b=0}^1\sum_{m=2}^n \abs{\tr{M((\rho_b^{(m)} - \rho_b^{(m-1)}\otimes \tau_d) \otimes \tau_d^{\otimes (n-m)})}}.
\end{align}
To simplify further, note that partial transposition of any subset of indices $S \subseteq [n]$ does not change the inner product of two operators $\tr{X^{\Gamma_S} Y^{\Gamma_S}}=\tr{XY}$. Further, note that if $0 \preceq E \preceq \Id$ is an effect and $X$ is traceless, $\tr{EX} \leq \frac{1}{2}\|X\|_1$. We apply these identities after fixing $m$ in the sum above and partially transpose just the $m$-th leg of the summand to obtain
\begin{align}
     \abs{\tr{M(\rho_0-\rho_1)}} \leq \sum_{b=0}^1\sum_{m=2}^n \frac{1}{2}\|((\rho_b^{(m)})^{\Gamma_m} - \rho_b^{(m-1)}\otimes \tau_d)\|_1,
\end{align}
where we have used the fact that the PPT-BOTH condition ensures $0 \preceq M^{\Gamma_m} \preceq \Id$ is a valid effect to be able to apply the trace norm bound. This is the essential proof technique that is common to all of our main results: telescoping sum, invariance of the Hilbert-Schmidt inner product under partial transposition, and an elementary upper bound on the overlap between effects and traceless Hermitian operators. Given this, we may summarize our main results and provide a high-level overview of the remaining proof details.

\begin{theorem}[Symmetric Werner State Indistinguishiability -- Informal, see Thm.~\ref{thm:symmetric-Werner-distinguishing}] Any PPT-BOTH measurement $\{M,\Id-M\}$ used to distinguish symmetric Werner states $\rho_0, \rho_1$ will achieve bias at most
\begin{align}
    \abs{\tr{M(\rho_0-\rho_1)}} \leq \frac{3}{2}  \frac{n(n-1)}{d}.
\end{align}
\end{theorem}
Symmetric Werner states, defined properly below, are Werner states are also satisfy $[P_d(\pi),\rho]=0$ for all permutations $\pi \in \SC_n$. With both of these symmetries, Schur-Weyl duality and Schur's lemma imply that $\rho$, along with all of its marginals, can be expressed as a convex combination of even simpler states. Doing so and using the convexity of the trace norm, results in the need to bound the trace norm of a partially transposed operator that has so much structure we are able to exactly compute the spectrum (see Prop.~\ref{prop:exact-spectrum}). Once we have determined this, simple bounds from standard representation theory of the symmetric group yield the desired upper bound.

This result has several implications, the first of which has to do with quantum property testing~\cite{montanaro2016Survey}. Specifically, the above result will allow us to prove a non-trivial lower bound on any PPT-BOTH tester for a unitarily invariant property.

\begin{corollary}[Testing Unitarily Invariant Properties -- Informal, see Cor.~\ref{cor:testing-unitarily-invariant-properties}] Any PPT-BOTH tester for a non-trivial unitarily invariant property requires $\Omega(\sqrt{d})$ samples to achieve constant bias.
\end{corollary}

To understand the idea here, consider the concrete task of purity testing.  With equal probability, you are given either a pure state or a state which is $\eps$-far from all pure states. The goal is to determine, with probability of success at least $2/3$, which is the case. The standard reduction from distinguishing to testing asks one to determine whether they were given a Haar random pure state or the maximally mixed state. In our framework, this is equivalent to distinguishing between the maximally mixed state on the symmetric subspace and the maximally mixed state on the whole space, i.e.
\begin{align}
   \sigma^{(n)} \coloneq \frac{\Pi_{\rm sym}}{\tr{\Pi_{\rm sym}}} \quad \text{versus} \quad \tau_d^{\otimes n} = \left(\frac{\Id}{d}\right)^{\otimes n}.
\end{align}
The assumption that a PPT-BOTH distinguisher $\{M,\Id-M\}$ exists corresponds to a constant lower bound on the bias $\abs{\tr{M(\sigma^{(n)} - \tau_d^{\otimes n})}}$. Together with the upper bound in the theorem above, we have
\begin{align}
    \Omega(1) \leq \abs{\tr{M(\sigma^{(n)} - \tau_d^{\otimes n})}} \leq O\left(\frac{n^2}{d}\right) \implies n \geq \Omega(\sqrt{d}).
\end{align}
This recovers the known purity testing lower bound from Ref.~\cite{chen2022Exponential} which utilizes the learning tree formalism to model potentially adaptive single-copy algorithms. In a recent work by the first two authors of the present paper, it was shown that this purity testing bound can, remarkably, be directly proved using almost know representation theory at all~\cite{akresh2026Optimal}. Our Theorem~\ref{thm:symmetric-Werner-distinguishing}, while being more involved to prove, is far more general, implying a non-trivial lower bound for any unitarily invariant property. 

As another example, consider the problem of low-rank testing~\cite{childs2007Weak,odonnell2015Quantum}. Given copies of an unknown state $\rho$ on $\mathbb C^d$,the task is to distinguish
$\operatorname{rank}(\rho)\le r$ from
\begin{align}
 \inf_{\operatorname{rank}(\sigma)\le r}
 \frac12\|\rho-\sigma\|_1 \ge \varepsilon.
\end{align}
For $d\ge2r$ and $0<\varepsilon\le1/2$, consider states
that are maximally mixed on Haar-random subspaces of
dimensions $r$ and $2r$, respectively. Writing $P_s$ for
a fixed rank-$s$ orthogonal projector, the corresponding
averaged $n$-copy states are
\begin{align}
 \rho_{\mathrm{yes}}^{(n)}
 &= \int_{\UC_d}
 \left(\frac{UP_rU^\dagger}{r}\right)^{\otimes n}\,dU
 \quad \text{and} \quad
 \rho_{\mathrm{no}}^{(n)}
 = \int_{\UC_d}
 \left(\frac{UP_{2r}U^\dagger}{2r}\right)^{\otimes n}\,dU.
\end{align}
Here $U$ is drawn once and the resulting state is copied
$n$ times. Every state in the yes ensemble has rank $r$,
whereas every state in the no ensemble is trace distance
$1/2$ from the set of states of rank at most $r$.
Both averaged states are symmetric multipartite Werner
states: Haar averaging gives invariance under collective
unitaries, and each tensor power is invariant under
permutations of the copies. With unrestricted measurement access, O'Donnell and Wright give a tester using $O(r^2)$ copies for constant $\eps$~\cite{odonnell2015Quantum}.
The above corollary implies that $\Omega(\sqrt{d})$ copies
are required under PPT-BOTH measurements, strengthening the best-known single-copy lower bound (cf. Theorem 40 in Ref.~\cite{ye2025Exponential}). Thus, for
$d\gg r^4$, these bounds establish an asymptotic sample
complexity separation between unrestricted and PPT-BOTH
measurements. In particular, for fixed $r$, the collective
upper bound is independent of $d$, whereas the PPT-BOTH
lower bound grows as $\sqrt{d}$.

We now turn to our main result in which we settle the performance of the Eggeling-Werner data hiding scheme introduced over two decades ago.
\begin{theorem}[Werner State Indistinguishiability -- Informal, see Thm.~\ref{thm:multipartite-werner-indistinguishability}] Any PPT-BOTH measurement $\{M,\Id-M\}$ used to distinguish Werner states $\rho_0, \rho_1$ will achieve bias at most
\begin{align}
    \abs{\tr{M(\rho_0-\rho_1)}} \leq \frac{3}{2}  \frac{n(n-1)}{d}.
\end{align}
\end{theorem}
The proof uses the above telescoping argument to reduce the distinguishing bias to trace norms of partially transposed operators. Mixed Schur–Weyl duality separates these operators into polynomial and non-polynomial sectors. An operator-order inequality controls all coherences in the polynomial sector simultaneously, while the non-polynomial blocks can be evaluated explicitly.

At any fixed security threshold $\varepsilon>0$, the bound guarantees bias at most $\varepsilon$ for every pair when $n\leq c \sqrt{\varepsilon d}$, for suitabyl chosen constant $c$. As pointed out in Ref.~\cite{harrow2023Approximate}, Schur–Weyl duality and a dimension counting argument imply there exist $N=\sum_{\lambda\vdash n}f_\lambda=2^{\Theta(n\log n)}$ mutually orthogonal Werner states, giving $\binom{N}{2}$ hiding pairs within this family. Taking $n=\Theta(\sqrt d)$ therefore certifies $2^{\Theta(\sqrt d\log d)}$ globally distinguishable messages at fixed security. The following proposition implies that the uniform security guarantee cannot have better asymptotic scaling, although particular pairs may remain secure well beyond this regime (as we discuss below).

\begin{proposition}[Tightness of Bounds -- Informal, see Prop.~\ref{prop:haar-vs-max-mixed-distinguisher}] 
    There exists a non-adaptive product measurement $\{M,\Id-M\}$ which distinguishes Werner states $\sigma^{(n)}$ and $\tau_d^{\otimes n}$ with bias 
    \begin{align}
        \tr{M(\sigma^{(n)} - \tau_d^{\otimes n}) } \geq \Omega\left(\frac{n^2}{d}\right).
    \end{align}
    Thus, the bound in Theorem~\ref{thm:multipartite-werner-indistinguishability} is asymptotically tight.
\end{proposition}
This not only establishes the worst-case optimality of our uniform upper bound on the bias, but it shows that, remarkably, the PPT-BOTH relaxation is tight. That is, even if we could analyze an exact restriction to LOCC, the existence of the above LOCC distinguisher with bias $\Omega(n^2/d)$ implies that one could not hope to tighten the upper bound uniformly. This is surprising because the PPT relaxation can lead to drastically looser results, compared to LOCC, in various quantum information processing tasks~\cite{beigi2010Approximating,bennett1999Quantum,cheng2023Discrimination,liu2023Complexity}.

We note that, although our bound is worst-case optimal, particular pairs of Werner states can be substantially harder to distinguish. For even $n$, Lancien and Winter~\cite{lancien2013Distinguishing} construct orthogonal Werner states whose PPT-BOTH distinguishing bias is $O(d^{-n/2})$, with optimal dimension dependence up to factors depending on $n$. Thus, worst-case optimality leaves room for much stronger hiding in particular constructions, highlighting the substantial variation in distinguishability within the family of multipartite Werner states. 

Together, our results establish the optimal uniform security scaling of multipartite Werner states and imply sample complexity lower bounds for testing unitarily invariant properties quantum under restricted measurements. The proof provides a common framework for both applications, combining a telescoping reduction with mixed Schur–Weyl duality to control trace norms of partially transposed operators. We hope this framework opens a fruitful route for applying representation-theoretic tools developed for port-based teleportation and other applications in quantum information theory to the broader study of the power and limitations of restricted-measurement protocols.

\section{Preliminaries}
\subsection{Essential Notation, Definitions, and Elementary Results}
Let us first introduce some notation and basic facts. 
We denote the algebra of linear operators on a Hilbert space $\cH$ by $\cL(\cH)$. Generally, we work in the Hilbert space $(\C^d)^{\otimes n}$. Let $\tau_d \coloneqq \Id_d/d$ be the maximally mixed state on $\C^d$.
For two Hermitian operators $X,Y$ we write $X\preceq Y$ iff $Y-X\succeq 0$, that is, $Y-X$ is positive semidefinite.

The class of Local Operations and Classical Communication (LOCC), while being well-motivated experimentally, is notoriously difficult to work with mathematically~\cite{chitambar2014Everything}. LOCC is strictly contained in the set of separable measurements (SEP), which admit a simpler mathematical form and, in turn, are contained within the class of all PPT measurements~\cite{chitambar2014Everything,matthews2009Distinguishability,lancien2013Distinguishing}. We have the following strict relationships between these classes
\begin{align}
    \operatorname{LOCC} \subsetneq \operatorname{SEP} \subsetneq \operatorname{PPT}.
\end{align}
Thus, allowing for PPT measurements is genuine relaxation of LOCC. To understand how the PPT relaxation yields mathematically tractable constraints, let us formally define partial transpositon.
\begin{definition}[Partial Transposition and PPT Condition]
\label{def:partial-transpose-ppt-multioutcome}
    Consider an operator $M \in \LC((\mathbb{C}^d)^{\otimes n})$
    and a subset of qudits $S \subseteq [n]$. Fix an orthonormal
    basis on each subsystem and write
    \begin{align}
        M = \sum_{\boldsymbol{i},\boldsymbol{j}\in[d]^n}
        M_{\boldsymbol{i},\boldsymbol{j}}
        \bigotimes_{k=1}^n \ketbra{i_k}{j_k}^{(k)},
        \qquad
        M_{\boldsymbol{i},\boldsymbol{j}}
        = \bra{\boldsymbol{i}}M\ket{\boldsymbol{j}},
    \end{align}
    where the superscript $(k)$ denotes the $k$-th subsystem.
    The partial transpose of $M$ on subsystems $S$, denoted
    $M^{\Gamma_S}$, transposes those tensor factors and leaves
    the others unchanged:
    \begin{align}
        M^{\Gamma_S}
        = \sum_{\boldsymbol{i},\boldsymbol{j}\in[d]^n}
        M_{\boldsymbol{i},\boldsymbol{j}}
        \bigotimes_{k=1}^n
        \begin{cases}
            \ketbra{j_k}{i_k}^{(k)}, & k\in S,\\
            \ketbra{i_k}{j_k}^{(k)}, & k\notin S.
        \end{cases}
    \end{align}
    We say $M$ is PPT if
    \begin{align}
        M^{\Gamma_S}\succeq0
        \quad \text{for every } S\subseteq[n].
    \end{align}
    In other words, $M$ remains positive semidefinite under
    partial transposition across every bipartition of
    $(\mathbb{C}^d)^{\otimes n}$.
\end{definition}
The partial transposition operation is mathematically easy to work with because it satisfies a few simple, but important properties.
\begin{proposition}[Partial Transposition Properties]\label{prop:partial-transpose-properties}
    Let $X$ and $Y$ be arbitrary operators on $(\C^d)^{\otimes n}$ and $S \subseteq [n]$. The following statements about the partial transpose are all true. \begin{enumerate}
        \item \textit{Linearity.} $(c_1X + c_2 Y)^{\Gamma_S} = c_1X^{\Gamma_S} + c_2Y^{\Gamma_S}$ for constants $c_1$ and $c_2$.
        \item \textit{Diagonal Operator Invariance.} Diagonal operators are invariant, in particular, $\Id^{\Gamma_S} = \Id$.
        \item \textit{Trace Invariance.} $\tr{X} = \tr{X^{\Gamma_S}}$
        \item \textit{Hilbert-Schmidt Invariance.} $\tr{XY} = \tr{X^{\Gamma_S}Y^{\Gamma_S}}$.
    \end{enumerate}
\end{proposition}

When considering binary distinguishing tasks, it suffices to consider two-outcome POVMs $\{M, \Id - M\}$ on $(\mathbb{C}^d)^{\otimes n}$, where we adopt the convention of $M$ corresponding to the guess ``$\rho_0$" (thus, $\Id-M$ corresponds to ``$\rho_1$."
\begin{definition}[PPT-BOTH Condition]\label{def:PPT-both-condition}
    A two-outcome POVM $\{M, \Id - M\}$ is PPT-BOTH if $M$ and $\Id -M$ are PPT. Equivalently, it is PPT-BOTH if \begin{align}
        0 \preceq M^{\Gamma_S} \preceq \Id \quad \forall S \subseteq [n].
    \end{align}
\end{definition}
In fact, restricting to binary POVMs doesn't affect the strict containment of the measurement classes (see, for example, Appendix C of Ref.\cite{harrow_testing_2010})
\begin{align}
    \operatorname{LOCC} \subsetneq \operatorname{SEP-BOTH} \subsetneq \operatorname{PPT-BOTH}.
\end{align}
This has the following essential implication for our work. Treating each copy as a separate party, any adaptive single-copy learning/testing protocol can be viewed as one-way LOCC measurement $\{M,\Id-M\}$ on $\HC^{\otimes n}$ across the partition $\HC_1:\cdots:\HC_n$. Thus, our PPT-BOTH bounds will immediately imply bounds for the experimentally-motivated adaptive single-copy setting.
 
Many of our main results require the use of various operator norms, which we now define.
\begin{definition}[Schatten $p$-norms]
    Let $1\leq p<\infty$ and $X$ be an $m\times m$ operator.
    Denote the singular values of $X$ by
    $\sigma_1\geq\sigma_2\geq\cdots\geq\sigma_m\geq0$.
    The Schatten $p$-norm is defined by
    \begin{align}
        \|X\|_p \coloneq
        \left(\sum_{i=1}^m \sigma_i^p\right)^{1/p}.
    \end{align}
    As $p\to\infty$, this norm approaches the maximum singular
    value, so we define $\|X\|_\infty\coloneq\sigma_1$.
\end{definition}
The Schatten norms are unitarily invariant and multiplicative
under tensor products:
\begin{align}
    \|UXV\|_p &= \|X\|_p \quad \text{and} \quad
    \|X\otimes Y\|_p = \|X\|_p\|Y\|_p,
\end{align}
where $U$ and $V$ are unitaries and $1\leq p\leq\infty$.
The norm $\|\cdot\|_\infty$ is called the operator norm or
spectral norm, while $\|\cdot\|_1$ is called the trace norm. The trace norm is additive over orthogonal direct sums:
\begin{align}
    \left\|\bigoplus_j X_j\right\|_1
    = \sum_j\|X_j\|_1.
\end{align}
For $X\succeq0$, its singular values equal its eigenvalues,
so $\|X\|_\infty$ is its largest eigenvalue and
$\|X\|_1=\tr{X}$. In particular, every density operator
$\rho$ satisfies $\|\rho\|_1=1$. Next, prove a small lemma that is essentially a tightening of H\"older's inequality, $\abs{\tr{XY}} \leq \|X\|_\infty \|Y\|_1$, possible when $X$ is a valid effect and $Y$ is traceless Hermitian.

\begin{lemma}\label{lem:overlap-to-one-norm-bound}
    If $X$ is an effect $0 \preceq X \preceq \mathbb{I}$, and $Y$ is traceless and Hermitian, then $\tr{XY} \leq \frac{1}{2} \|Y\|_1$.
\end{lemma}
\begin{proof}
    Let $Y = Y_+ - Y_-$, where $Y_+$ and $Y_- $ are positive operators acting on the positive and negative eigenspaces, respectively. Note that $\tr{Y}=0 \implies \tr{Y_+} = \tr{Y_-} = \frac{1}{2} \|Y\|_1.$ We may then write
    \begin{align}
        \tr{XY} &= \tr{XY_+} - \tr{XY_-},\\
        &\leq \tr{XY_+}, \quad & 0 \preceq Y_- \implies \tr{XY_-} \geq 0\\
        &\leq \tr{Y_+}, \quad &X \preceq \mathbb{I}\\
        &= \frac{1}{2} \|Y\|_1,
    \end{align}
    as desired.
\end{proof}

\subsection{Standard Schur-Weyl Duality}\label{sec:schur-weyl}

Let $\SC_n$ denote the permutation group on $n$ objects. 
We consider the representation on $(\C^d)^{\otimes n}$ by permuting systems via the unitary permutation operators 
\begin{align}
\pi\mapsto P_d(\pi)= \sum_{i_1,\dots,i_n \in [d]}\ketbra{i_{\pi^{-1}(1)},\dots,i_{\pi^{-1}(n)}}{i_1,\dots,i_n}.
\label{eq:Sn-representation}
\end{align}

We use $\vee^n(\C^d) \leq (\C^d)^{\otimes n}$ to denote the \emph{symmetric subspace}, defined as the subspace invariant under the action of all permutations: 
\begin{align} 
    \vee^n(\C^d)=\lbrace |v\rangle\in (\C^d)^{\otimes n} : P_d(\pi)|v\rangle = |v\rangle \text{ for all $\pi\in S_d$}\rbrace
\end{align}
That subspace has a corresponding orthogonal projector that we denote by $\Pi_{\text{sym}}^{d,n}$, or simpy $\Pi_n$ when the local dimension $d$ is clear from context.
It is given by the expression
\begin{align}
    \Pi_{\text{sym}}^{d,n}  = \frac{1}{n!}\sum_{\pi \in \SC_n}P_d(\pi),
    \label{eq:symmetric-subspace-projector}
\end{align}
and the dimension of the symmetric subspace is equal to
\begin{align}
    d[n]\coloneqq \dim \vee^n(\C^d) = \trc \Pi_{\mathrm{sym}}^{d,n} = \binom{d+n-1}{n}.
\end{align}

Along with the symmetric group $\SC_n$, we also consider the representation on $(\C^d)^{\otimes n}$ of the group $\mathcal{U}_d$ of $d\times d$ unitary matrices, acting as
\begin{align}
    U \mapsto U^{\otimes n}.
    \label{eq:Ud-representation}
\end{align} 
As a compact group, $\UC_d$ possesses a unique probability measure called the Haar measure, which is invariant under left- and right-translations by a fixed unitary~\cite{mele2024Introduction}. 

We define the $n$-fold twirling channel on $(\mathbb{C}^d)^{\otimes n}$ as the average over conjugating by unitaries with respect to the Haar measure:
\begin{align}
    \TC_n(X) = \underset{U \in \mathcal{U}_d}{\mathbb{E}}[U^{\otimes n} X(U^\dagger)^{\otimes n}],
\end{align}
where $\mathbb{E}$ denotes the expectation with respect to the Haar measure.
Notably, the resulting operator from twirling a pure state is well-known.
\begin{lemma}[Twirling $n$ Copies of a Pure State, Prop.~6~\cite{harrow2013Church}] \label{lem:haar-twirl-symmetric-subspace}
For any $|\phi\rangle\in\mathbb{C}^d$,
    \begin{align}
        \TC_n(\ketbra{\phi}{\phi}^{\otimes n}) = \underset{U \in \UC_d}{\mathbb{E}}\left[U^{\otimes n} \ketbra{\phi}{\phi}^{\otimes n}(U^\dagger)^{\otimes n}\right] = \frac{\Pi_n}{\tr{\Pi_n}} = \frac{1}{d[n]} \Pi_n.
    \end{align}
    In words, the density matrix corresponding to $n$ copies of an unknown Haar-random pure state is the maximally mixed state on the symmetric subspace of $(\mathbb{C}^d)^{\otimes n}$.
\end{lemma}

Lemma \ref{lem:haar-twirl-symmetric-subspace} is a special case of a powerful duality known as Schur-Weyl duality \cite{goodman2009symmetry,etingof2011introduction,fulton2013representation}.
The starting point of this duality is to note that the two representations \eqref{eq:Sn-representation} of $\SC_n$ and \eqref{eq:Ud-representation} of $\cU_d$ commute with each other. 
Recall the definition of the commutant of $\SC_n$ in $\cL((\mathbb{C}^d)^{\otimes n})$ as the algebra of all operators $X\in \cL((\mathbb{C}^d)^{\otimes n})$ satisfying $[X,P_d(\pi)]=0$ for all $\pi\in \SC_n$.
The above statement then means that the commutant of $\SC_n$ \emph{contains} all $U^{\otimes n}$ for $U\in\cU_d$, and vice versa.
Succinctly, Schur-Weyl duality states that the two representations \eqref{eq:Sn-representation} and \eqref{eq:Ud-representation} \emph{generate} each other's commutant.
In other words, any $X\in \cL((\mathbb{C}^d)^{\otimes n})$ satisfying $[X,P_d(\pi)]=0$ for all $\pi\in \SC_n$ can be written as a linear combination of operators $U^{\otimes n}$ for $U\in \cU_d$, and a similar statement holds for any $Y\in \cL((\mathbb{C}^d)^{\otimes n})$ with $[Y,U^{\otimes n}]=0$ for all $U\in\cU_d$.

There is an equivalent, and for our purposes quite useful, formulation of Schur-Weyl duality in terms of a decomposition of the representation space $(\mathbb{C}^d)^{\otimes n}$:
\begin{align}
    (\mathbb{C}^d)^{\otimes n} &\overset{\UC_d \times \SC_n}{\cong} \bigoplus_{\lambda \vdash_d n} \QC_{\lambda}^d \otimes \PC_{\lambda}, \label{eq:schur-weyl-decomposition}\\
    U^{\otimes n} &\overset{\phantom{\UC_d \times \SC_n}}{\cong} \bigoplus_{\lambda \vdash_d n} q_\lambda(U) \otimes \Id_{\PC_{\lambda}}\\
    P_d(\lambda) &\overset{\phantom{\UC_d \times \SC_n}}{\cong} \bigoplus_{\lambda \vdash_d n} \Id_{\QC_{\lambda}^d} \otimes R_{\lambda} (\pi),
\end{align}
where $\lambda \vdash_d n$ denotes a partition of $n$ with at most $d$ rows, $(\QC_{\lambda}^d,q_\lambda)$ is an irreducible representation (irrep) of $\UC_d$, and $(\PC_{\lambda},R_\lambda)$ is an irrep of $\SC_n$~\cite{fulton2013representation}.
The spaces $\cQ_\lambda^d$ and $\cP_\lambda$ are called Weyl and Specht modules, respectively.
We will denote the dimensions of these spaces as $  q_{\lambda} = \dim{\QC_{\lambda}^d}$ and $f_{\lambda} = \dim{\PC_{\lambda}}$, respectively.\footnote{The dimension $q_\lambda$ depends on the local dimension $d$, but we suppress this dependence for the sake of readability.}

Next, let $V_\lambda\colon\QC_\lambda^d\otimes\PC_\lambda
\to(\mathbb C^d)^{\otimes n}$
be the isometry whose columns are the `Schur basis' vectors in the $\lambda$-block (used in \eqref{eq:schur-weyl-decomposition}) written in the standard basis. Thus $V_\lambda$ converts coefficients within that block into a vector in the physical tensor-product space. The projector onto the $\lambda$-isotypic subspace $\QC_\lambda^d\otimes\PC_\lambda$ is given by $\Pi_\lambda=V_\lambda V_\lambda^\dagger,$
which are mutually orthogonal, $\Pi_{\lambda} \Pi_{\lambda'} = \delta_{\lambda \lambda'} \Pi_{\lambda}$, and resolve the identity on the full space, 
\begin{align} 
    \sum_{\lambda \vdash_d n} \Pi_{\lambda} = \Id_{(\mathbb{C}^d)^{\otimes n}}.
\end{align}
Furthermore, as isotypical projectors the $\Pi_\lambda$ are invariant under both representations $\SC_n\ni\pi\mapsto P_d(\pi)$ and $\cU_d\ni U\mapsto U^{\otimes n}$.
We will later deal with these projectors on systems with different numbers of subsystems, and thus also use the notation $\Pi_{\lambda}^{(n)}$ to keep track of those.
For the one-row diagram $\lambda = (n,0,\dots,0)\equiv(n)$, the isotypical projector $\Pi_{(n)}$ is exactly the projector $\Pi_{\text{sym}}^{d,n}$ in \eqref{eq:symmetric-subspace-projector}.

The Schur-Weyl decomposition \eqref{eq:schur-weyl-decomposition} gives a succinct way to describe invariant objects via Schur's Lemma, which states that any linear map commuting with two irreducible representations of a group $G$ is either identically 0 if the two irreps are inequivalent, or proportional to the identity otherwise.
For example, any permutation-invariant operator $X\in\cL((\mathbb{C}^d)^{\otimes n})$ satisfying $[X,P_d(\pi)]=0$ for all $\pi\in \SC_n$ is of the form
\begin{align}
    X &\cong \bigoplus_{\lambda\vdash_d n} X_\lambda \otimes \Id_{\cP_\lambda}
\intertext{for some operators $X_\lambda\in \cL(\cQ_\lambda^d)$.
Likewise, any operator $Y\in\cL((\mathbb{C}^d)^{\otimes n})$ satisfying $[Y,U^{\otimes n}]=0$ for all $U\in\cU_d$ is of the form}
    Y &\cong \bigoplus_{\lambda\vdash_d n} \Id_{\cQ_\lambda^d} \otimes Y_\lambda
\end{align}
for some operators $Y_\lambda\in\cL(\cP_\lambda)$.
Lemma \ref{lem:haar-twirl-symmetric-subspace} is a special case of this observation: a product state $\ketbra{\phi}{\phi}^{\otimes n}$ is fully supported on the symmetric subspace $\vee^n(\mathbb{C}^d) = \cQ_{(n)}$, and $\dim\cP_{(n)}=1$, so that the twirling with respect to $U^{\otimes n}$ trivializes the operator $Y_{(n)}$ to a multiple of the identity $\Id_{\cQ_{(n)}^d} = \Pi_{\mathrm{sym}}^{d,n}$.
Taking both symmetries together, an operator $Z$ invariant under both $P_d(\pi)$ for all $\pi\in \SC_n$ and $U^{\otimes n}$ for all $U\in\cU_d$ is of the form
\begin{align}
    Z \cong \bigoplus_{\lambda\vdash_d n} z_\lambda \Id_{\cQ_{\lambda}^d}\otimes \Id_{\cP_\lambda} = \sum_{\lambda\vdash_d n} z_\lambda \Pi_\lambda.
\end{align}
In the quantum information theory literature, quantum states with $U^{\otimes n}$-symmetry are called \emph{multipartite Werner states}, and they are called \emph{multipartite symmetric Werner states} if they are additionally permutation-invariant~\cite{werner1989Quantum}.

\begin{definition}[Werner states]
\label{def:Werner-states}
A density operator $\rho$ on $(\mathbb C^d)^{\otimes n}$
is a \textit{multipartite Werner state} if
$[\rho,U^{\otimes n}]=0$ for every $U\in\UC_d$.
By Schur--Weyl duality, these are precisely the states
whose Schur-basis form is
\begin{align}
 \rho =
 \bigoplus_{\lambda\vdash_d n}
 p_\lambda
 \frac{\Id_{\QC_\lambda^d}}{q_\lambda}\otimes\sigma_\lambda,
\end{align}
where $\{p_\lambda\}_\lambda$ is a probability distribution
and each $\sigma_\lambda$ is a density operator on $\PC_\lambda$. If additionally $[\rho,P_d(\pi)]=0$ for every $\pi\in\SC_n$,
we call $\rho$ a \textit{symmetric multipartite Werner state}.
In this case, every block with nonzero weight has
$\sigma_\lambda=\Id_{\PC_\lambda}/f_\lambda$, so
\begin{align}
 \rho=\sum_{\lambda\vdash_d n}p_\lambda\omega_\lambda,
 \qquad
 \omega_\lambda=\frac{\Pi_\lambda}{q_\lambda f_\lambda}.
\end{align}
Thus $\omega_\lambda$ is the maximally mixed state
on the entire $\lambda$-isotypic subspace.
\end{definition}
While this block diagonal form of multipartite Werner states is already quite useful for us, it will also come in handy to have an explicit decomposition in terms of permutation operators. The following proposition gives just such an expansion. It follows from Schur--Weyl duality and matrix-coefficient orthogonality; we include a proof in Appendix~\ref{app:technical-lemma-proofs}.

\begin{proposition}[Fourier Expansion of Werner States] \label{prop:fourier-decomp-Werner-states}
    Let $\rho$ be a multipartite Werner state on $(\mathbb{C}^d)^{\otimes n}$. Then, $\rho$ lies in the linear span of permutation operators and admits the expansion
    \begin{align}
        \rho &= \frac{1}{n!} \sum_{\pi \in \SC_n} \left(\sum_{\lambda \vdash_d n} \frac{p_{\lambda} f_{\lambda}}{q_{\lambda}} \tr{\sigma_{\lambda} R_{\lambda}(\pi^{-1})}\right)P_d(\pi).
    \end{align}
    If $\rho$ is a symmetric multipartite Werner state, this expression simplifies to 
    \begin{align}
        \rho &= \frac{1}{n!} \sum_{\pi \in \SC_n} \left(\frac{p_{\lambda}}{q_{\lambda}} \chi_{\lambda}(\pi^{-1})\right) P_d(\pi),
    \end{align}
    where $\chi_{\lambda}(\pi) = \tr{R_{\lambda}(\pi)}$ is the character value of $\pi$.
\end{proposition}

Throughout this work, we will need to work with an explicit form for the marginals of Werner states. Before deriving the form of the marginals, it will help to review some additional facts from the representation theory of the symmetric group.
\subsubsection{Branching Rules and Young Tableaux Identities} \label{sec:branching-rules-young-identities}
Consider operators on $(\mathbb{C}^d)^{\otimes m}$, where $2 \leq m \leq n$. Throughout this work, we denote $\alpha \vdash_d m-2$, $\beta \vdash_d m-1$, and $\lambda \vdash_d m$. If $\beta$ is obtained from $\lambda$ by removing one valid box from the Young diagram, we will write $\beta \prec \lambda$. Viewing $\SC_{m-1}$ as the subgroup of $S_m$ fixing $m$, the Specht branching rule~\cite{fulton2013representation} gives
\begin{align}
    \mathrm{Res}_{\SC_{m-1}}^{\SC_m} \PC_{\lambda} \cong \bigoplus_{\beta \prec \lambda} \PC_{\beta}.
    \label{eq:branching-rule}
\end{align}
In words, when the $m$-th site is fixed, the Specht module $\PC_{\lambda}$ splits into irreducible blocks, each occurring once and indexed by Young diagrams obtained by removing one valid box from $\lambda$. Note that this multiplicity-free branching rule implies the following relationship between irrep dimensions $f_{\lambda} = \sum_{\beta \prec \lambda} f_{\beta}.$

We will also need a closed-form relationship between the Weyl modules before and after this restriction. To this end, let $u =(i,j)$ be a box in a Young diagram and define its \textit{content} by $c(u) \coloneq j-i$. Letting the \textit{hook length} for a given box $u$ within a valid Young diagram $\lambda \vdash_d m$ be denoted as $h_{\lambda}(u)$, we may write the dimensions of the Specht and Weyl modules as  
\begin{align}
    f_{\lambda} = \frac{m!}{\prod_{u \in \lambda} h_{\lambda}(u)} \quad \text{and} \quad q_{\lambda} = \prod_{u \in \lambda} \frac{d + c(u)}{h_{\lambda}(u)}.
\end{align}
These so-called hook-length and hook-content formulas are standard in the literature, but we refer the reader to Ref.~\cite{leditzky2026Representation} for a pedagogical treatment with examples. Throughout the proofs of our main results, we will make use of the following bounds which follow from taking ratios of the hook-length and hook-content formulas, simplifying, and taking an elementary upper bound.

\begin{lemma}[One-box Dimension Bounds]\label{lem:one-box-dimension-bounds} Let $m \geq 2$ and $\lambda \vdash_d m$. Consider Young diagrams satisfying $\alpha \prec \beta \prec \lambda$. The following bounds then hold:
\begin{enumerate}
    \item If $d > m-1$, then 
    \begin{align}
        \sum_{\beta \prec \lambda} \frac{q_{\beta}}{q_{\lambda}} \leq \frac{m}{d-m+1}.
    \end{align}
    \item If $d > m-2$, then,
    \begin{align}
        \sum_{\alpha \prec \beta} \frac{q_{\alpha}}{q_{\beta}} \leq \frac{m-1}{d-m+2}.
    \end{align}
\end{enumerate}
\end{lemma}
\begin{proof}
    Both identities follow from the observation that a box in a Young diagram with $r$ boxes has content at least $-(r-1)$. Thus, in the first case, we have 
    \begin{align}
        \frac{1}{m} \sum_{\beta \prec \lambda} \frac{q_{\beta}}{q_{\lambda}} = \sum_{\beta \prec \lambda} \frac{f_{\beta}}{f_{\lambda}[d+c(\lambda \setminus \beta)]} \leq \frac{1}{d - m + 1} \sum_{\beta \prec \lambda} \frac{f_{\beta}}{f_{\lambda}} = \frac{1}{d-m+1},
    \end{align}
    and multiplying both sides by $m$ yields the desired result. In the last line we used $  f_{\lambda} = \sum_{\beta \prec \lambda} f_{\beta}.$ The second case is proved identically. Note that, when $m=2$, the latter result requires the definition $q_{\emptyset} = f_{\emptyset} =1.$
\end{proof}
With these identities in place, we are prepared to derive the explicit form of the marginals needed in the proofs of our main results.
\subsubsection{Marginals of Werner States}
 Fix $m \ge 2$ and a partition $\lambda\vdash_d m$ with predecessor $\beta \prec \lambda$. Let $J_{\lambda\to\beta}:\mathcal P_\lambda\to\mathcal P_\beta$ be a map satisfying
\begin{align}
J_{\lambda\to\beta}J_{\lambda\to\beta}^\dagger
 =\Id_{\mathcal P_\beta},
 \qquad
 \sum_{\beta\prec\lambda}
 J_{\lambda\to\beta}^\dagger J_{\lambda\to\beta}
 =\Id_{\mathcal P_\lambda}.
\end{align}
That is, this map extracts the coordinates belonging to the $\beta$-block and its adjoint places vectors back into that block, with zeros on all other blocks.  Finally, define the sub-normalized diagonal branch states and their weights by
\[
 \sigma_\beta
 =
 J_{\lambda\to\beta}\sigma J_{\lambda\to\beta}^\dagger,
 \qquad
 t_\beta=\operatorname{tr}\sigma_\beta.
\]
Then $\sigma_\beta\succeq0$ and
$\sum_{\beta\prec\lambda}t_\beta=1$.
The state $\sigma$ may also have off-diagonal branch blocks
$J_{\lambda\to\beta}\sigma J_{\lambda\to\gamma}^\dagger$
for $\beta\ne\gamma$; we impose no restriction on these. With this notation in place, we are ready to determine a form of the marginal states.

\begin{lemma}[Marginal of a Schur--Weyl block]
\label{lem:general-marginal}
With the notation above,
\begin{equation}
\label{eq:marginal-Schur-block}
 \operatorname{tr}_m\Omega_{\lambda,\sigma}
 \cong
 \bigoplus_{\beta\prec\lambda}
 \frac{I_{\mathcal Q_\beta^d}}{q_\beta}\otimes\sigma_\beta.
\end{equation}
Here the identification is the fixed Schur--Weyl
decomposition on the first $m-1$ subsystems, and the
blocks indexed by $\beta\not\prec\lambda$ are zero.
\end{lemma}

\begin{proof}
    We will show that, for all $X$, we have
    \begin{align}
        \trc\left(X 
        \trc_m \Omega_{\lambda,\sigma}\right) = \trc\left(X \left(\bigoplus\nolimits_{\beta\prec\lambda} q_\beta^{-1}\Id_{\cQ_\beta^d}\otimes \sigma_\beta \right)\right),
        \label{eq:to-prove}
    \end{align}
    from which the claim of the lemma follows.
    Since both $\trc_m \Omega_{\lambda,\sigma}$ and $\bigoplus\nolimits_{\beta\prec\lambda} q_\beta^{-1}\Id_{\cQ_\beta^d}\otimes \sigma_\beta$ are invariant under $U^{\otimes m-1}$ for all $U\in U(d)$, we can assume without loss of generality that $X$ in \eqref{eq:to-prove} is invariant too, and thus
    \begin{align}
        X \cong \bigoplus_{\alpha\vdash_d m-1} \Id_{\cQ_\alpha^d} \otimes X_\alpha
    \end{align}
    for some $X_\alpha$ supported on $\cP_\alpha$. For the right-hand side of \eqref{eq:to-prove}, the direct sum structures of $X$ and $\bigoplus\nolimits_{\beta\prec\lambda} q_\beta^{-1}\Id_{\cQ_\beta^d}\otimes \sigma_\beta$ imply that
    \begin{align}
        \trc\left(X \left(\bigoplus\nolimits_{\beta\prec\lambda} q_\beta^{-1}\Id_{\cQ_\beta^d}\otimes \sigma_\beta \right)\right) &= \sum_{\beta\prec\lambda} \frac{1}{q_\beta} \trc\left( \Id_{\cQ_\beta^d}\otimes X_\beta\sigma_\beta\right) = \sum_{\beta\prec\lambda} \trc(X_\beta\sigma_\beta).
    \end{align}
    For the left-hand side, we have
    \begin{align}
        \trc\left(X \trc_m \Omega_{\lambda,\sigma}\right) = \trc\left( (X\otimes \Id_m) \Omega_{\lambda,\sigma} \right) &= \frac{1}{q_\lambda} \trc\left[ \left(\left(\bigoplus\nolimits_{\alpha\vdash_d m-1}\Id_{\cQ_\alpha^d} \otimes X_\alpha\right) \otimes \Id_m\right) \left(\Id_{\cQ_\lambda^d} \otimes\sigma\right)\right]\\
        &= \frac{1}{q_\lambda} \trc\left[ \left(\bigoplus\nolimits_{\alpha\vdash_d m-1} \left(\Id_{\cQ_\alpha^d} \otimes \Id_m\right) \otimes  X_\alpha\right)  \left(\Id_{\cQ_\lambda^d} \otimes\sigma\right)\right]\\
        &= \frac{1}{q_\lambda} \trc\left[ \left(\bigoplus\nolimits_{\alpha\vdash_d m-1} \left(\Id_{\cQ_\alpha^d} \otimes \Id_{\cQ^d_{(1)}}\right) \otimes  X_\alpha\right)  \left(\Id_{\cQ_\lambda^d} \otimes\sigma\right)\right],
        \label{eq:LHS-trace}
    \end{align}
    where we interpreted $\Id_m\cong \cQ_{(1)}^d \otimes \cP_{(1)}\cong  \cQ_{(1)}^d$ as the $m=1$ case of Schur-Weyl duality.
    The Pieri rule yields the decomposition
    \begin{align}
        \cQ_{\alpha}^d \otimes \cQ_{(1)}^d \cong \sum_{\mu \succ \alpha} \cQ_{\mu}^d.
    \end{align}
    Using the corresponding identity for the projectors onto the $\cQ_*^d$ in \eqref{eq:LHS-trace}, 
    \begin{align}
        &\frac{1}{q_\lambda} \trc\left[ \left(\bigoplus\nolimits_{\alpha\vdash_d m-1} \left(\Id_{\cQ_\alpha^d} \otimes \Id_{\cQ^d_{(1)}}\right) \otimes  X_\alpha\right)  \left(\Id_{\cQ_\lambda^d} \otimes\sigma\right)\right] \\
        &= \frac{1}{q_\lambda} \trc\left[ \left(\bigoplus\nolimits_{\alpha\vdash_d m-1} \left(\sum\nolimits_{\mu\succ\alpha} \Id_\mu\right) \otimes  X_\alpha\right)  \left(\Id_{\cQ_\lambda^d} \otimes\sigma\right) \right]\\
        &= \frac{1}{q_\lambda} \trc\left[ \left(\bigoplus\nolimits_{\mu\vdash_d m} \Id_{\cQ_\mu^d} \otimes  \left(\bigoplus\nolimits_{\alpha\prec\mu} X_\alpha\right) \right)   \left(\Id_{\cQ_\lambda^d} \otimes\sigma\right) \right] \label{eq:reshuffled}\\
        &= \frac{1}{q_\lambda} \trc\left[ \Id_{\cQ_\lambda^d} \otimes \left(\bigoplus\nolimits_{\alpha\prec\lambda} X_\alpha\right) \sigma \right] \label{eq:sigma-blockform}\\
        &= \sum_{\alpha\prec\lambda} \trc(X_\alpha \sigma_\alpha),
    \end{align}
    where in \eqref{eq:reshuffled} we regrouped the direct summands, and in \eqref{eq:sigma-blockform} we used the block form $\sigma \cong (\sigma_{\beta,\beta'})_{\beta,\beta'\vdash_d m-1}$ of $\sigma$ with respect to the branching rule $\cP_\lambda \cong \bigoplus_{\beta\prec\lambda} \cP_\beta$ obtained by restricting the tensor representation of $S_m$ to $\SC_{m-1}$.
    With $\sigma_\beta\equiv \sigma_{\beta,\beta}$, this proves equality in \eqref{eq:to-prove}, and hence the claim.
\end{proof}

\begin{proposition}[Marginals of Werner States]
\label{prop:Werner-marginals}
Let $\rho$ be a multipartite Werner state on
$(\mathbb C^d)^{\otimes m}$, with Schur-basis form
\begin{align}
 \rho =
 \bigoplus_{\lambda\vdash_d m}
 p_\lambda\frac{\Id_{\QC_\lambda^d}}{q_\lambda}
 \otimes\sigma_\lambda.
\end{align}
Then, tracing out the $m$-th site yields a new Werner state on $(\mathbb{C}^d)^{\otimes (m-1)}$
\begin{align}
\label{eq:Werner-marginal-general}
 \operatorname{tr}_m\rho
 &=
 \bigoplus_{\beta\vdash_d m-1} \frac{1}{q_{\beta}}
 \Id_{\QC_\beta^d}\otimes \sigma_\beta \quad \text{where} \quad 
 {\sigma}_\beta =
 \sum_{\substack{\lambda\vdash_d m\\\beta\prec\lambda}}
 p_\lambda
 J_{\lambda\to\beta}\sigma_\lambda
 J_{\lambda\to\beta}^\dagger.
\end{align}
Here $\sigma_\beta\succeq0$ and
$\sum_\beta\operatorname{tr}S_\beta=1$. If $\rho$ is additionally permutation invariant, write
$\rho=\sum_{\lambda\vdash_d m}p_\lambda\omega_\lambda$,
where $\omega_\lambda=\Pi_\lambda/(q_\lambda f_\lambda)$.
Then
\begin{align}
\label{eq:Werner-marginal-symmetric}
 \operatorname{tr}_m\rho
 &=
 \sum_{\beta\vdash_d m-1}p'_\beta\omega_\beta \quad \text{where} \quad  \operatorname{tr}_m\omega_\lambda
 =
 \sum_{\beta\prec\lambda}
 \frac{f_\beta}{f_\lambda}\,\omega_\beta \quad \text{and} \quad
 p'_\beta
 =
 \sum_{\substack{\lambda\vdash_d m\\\beta\prec\lambda}}
 p_\lambda\frac{f_\beta}{f_\lambda}.
\end{align}

\end{proposition}

\begin{proof}
Applying Lemma~\ref{lem:general-marginal} to each
$\lambda$-block and collecting contributions to the same
$\beta$-block gives \eqref{eq:Werner-marginal-general}.
Each $S_\beta$ is positive semidefinite, and
\begin{align}
 \sum_\beta\operatorname{tr}S_\beta
 &=
 \sum_\lambda p_\lambda
 \operatorname{tr}\!\left[
 \sigma_\lambda
 \sum_{\beta\prec\lambda}
 J_{\lambda\to\beta}^\dagger J_{\lambda\to\beta}
 \right]
 =1.
\end{align}
In the symmetric case, we may take
$\sigma_\lambda=\Id_{\PC_\lambda}/f_\lambda$, so
\begin{align}
 J_{\lambda\to\beta}\sigma_\lambda
 J_{\lambda\to\beta}^\dagger
 =
 \frac{\Id_{\PC_\beta}}{f_\lambda} \quad \text{implying} \quad 
 S_\beta
 =
 \left(
 \sum_{\substack{\lambda\vdash_d m\\\beta\prec\lambda}}
 p_\lambda\frac{f_\beta}{f_\lambda}
 \right)
 \frac{\Id_{\PC_\beta}}{f_\beta},
\end{align}
which gives \eqref{eq:Werner-marginal-symmetric}.
The coefficients $p'_\beta$ form a probability distribution because $\sum_{\beta\prec\lambda}f_\beta=f_\lambda$.
\end{proof}

\subsection{Mixed Schur-Weyl Duality} \label{sec:mixed-Schur-Weyl}

The representation $U\mapsto U^{\otimes n}$ of the unitary group $\cU_d$ featuring in Schur-Weyl duality can be generalized to the following representation on $(\mathbb{C}^d)^{\otimes n}$:
\begin{align}
    \cU_d\ni U \mapsto U^{\otimes n-k} \otimes \overline{U}^{\otimes k},
    \label{eq:mixed-unitary-representation}
\end{align}
where $\overline{U}$ denotes the unitary obtained from $U$ by complex-conjugating each entry.
Describing the irreducible representations of \eqref{eq:mixed-unitary-representation} and determining the commutant of the algebra generated by it has been a recent topic of interest in quantum information theory known as `mixed Schur-Weyl duality', see for example \cite{Kopszak2021multiportbased,mozrzymas2021optimalmultiport,studzinski2022efficient,grinko2024linear,fei2023Efficient,nguyen2023mixed,grinko2023GelfandTsetlin}.

In this work, we focus on the case $k=1$ in \eqref{eq:mixed-unitary-representation}, which can be described using the standard representation theory of the symmetric and unitary groups discussed in Sec.~\ref{sec:schur-weyl}.
A well-known example of a task in which the symmetry $U^{\otimes n-1}\otimes \overline{U}$ is relevant is port-based teleportation \cite{ishizaka2008asymptotic,ishizaka2009quantum,studzinski2017port,mozrzymas2018optimal,christandl2021Asymptotic,leditzky2022Optimality,kim2026resourcetheory}, which features operators of the form 
\begin{align} 
    \sum_{i=1}^{n-1} \Id_{1}\otimes \dots\otimes \Id_{i-1}\otimes \Id_{i+1}\otimes \Id_{n-1} \otimes \psi^+_{i,n}.
    \label{eq:PBT-operator}
\end{align}
Here, $\psi^+$ denotes a maximally entangled state with $(U\otimes \overline{U})$-symmetry, which in turn equips the operator in \eqref{eq:PBT-operator} with a $U^{\otimes n-1}\otimes \overline{U}$-symmetry.

A second example, and the one relevant for the present work, is the discussion of partially transposed Werner states:
If $X\in\cL((\mathbb{C}^d)^{\otimes n})$ satisfies $U^{\otimes n} X (U^\dagger)^{\otimes n} = X$ for all $U\in \cU_d$, then also
\begin{align}
    X^{\Gamma_n} = \left(U^{\otimes n} X (U^\dagger)^{\otimes n}\right)^{\Gamma_n} = \left( U^{\otimes n-1}\otimes \overline{U}\right) X^{\Gamma_n} \left( U^{\otimes n-1}\otimes \overline{U}\right)^\dagger \quad \text{for all $U\in\cU_d$}
    \label{eq:partially-transposed-werner-state}
\end{align}
by elementary properties of the partial trace.

Partially transposed Werner states satisfying the relation in \eqref{eq:partially-transposed-werner-state} are crucial in this work.
Thus, in order to efficiently describe them, we are interested in finding a decomposition of the representation space $(\mathbb{C}^d)^{\otimes n}$ in terms of the irreps of the unitary group, in analogy to the Schur-Weyl decomposition \eqref{eq:schur-weyl-decomposition}.
This can be achieved using the following method employed in \cite[App.~A]{christandl2021Asymptotic} (see also \cite{leditzky2022Optimality,fei2023Efficient}).
One starts with considering $(\mathbb{C}^d)^{\otimes n}$ as a representation space for $\cU_d\times\cU_d$ acting as $U^{\otimes n-1}\otimes V$ and applies the Schur-Weyl decomposition to the first $n-1$ systems:
\begin{align} 
    (\mathbb{C}^d)^{\otimes n} &= (\mathbb{C}^d)^{\otimes n-1} \otimes (\mathbb{C}^d)^*\\
    &\cong \left( \bigoplus_{\beta\vdash_d n-1} \cQ_{\beta}^d\otimes \cP_\beta \right) \otimes (\mathbb{C}^d)^*,
    \label{eq:mixed-decomposition-primer}
\end{align} 
where $(\mathbb{C}^d)^*$ denotes the \emph{dual} representation of the defining (irreducible) representation $\mathbb{C}^d = \cQ_{(1)}^d$ of $\cU_d$.
Now, setting $U=V$, the tensor representation $\cQ_{\beta}^d \otimes (\mathbb{C}^d)^* \cong \cQ_{\beta}^d \otimes (\cQ_{(1)}^d)^*$ decomposes into $\cU_d$-irreps according to the dual Pieri rule \cite{fulton2013representation}:
\begin{align}
    \cQ_{\beta}^d \otimes (\cQ_{(1)}^d)^* = \bigoplus_{i\colon \beta_i>\beta_{i+1}} \cQ_{\alpha^{(i)}}, 
    \label{eq:dual-Pieri-rule}
\end{align}
where $\alpha^{(i)} = (\beta_1,\dots,\beta_{i-1},\beta_i-1,\beta_{i+1},\dots,\beta_d)$ is the weight obtained from $\beta$ by removing $1$ in position $i$ (provided $\beta_i>\beta_{i+1})$. 
Note that we set $\beta_{d+1}=-\infty$ so that $\alpha^{(d)}$ always appears in this direct sum. The new weight $\alpha^{(i)}$ is a partition (corresponding to a Young diagram) if $(\alpha^{(i)})_{d}\geq 0$.
Using \eqref{eq:dual-Pieri-rule} in \eqref{eq:mixed-decomposition-primer} now yields the desired decomposition of $(\mathbb{C}^d)^{\otimes m}$ as a representation of $\cU_d\times S_{m-1}$, where $U\in\cU_d$ acts as $U^{\otimes n-1}\otimes \overline{U}$, and $S_{m-1}$ acts on the first $n-1$ systems in $(\mathbb{C}^d)^{\otimes m}$:
\begin{align} \label{eq:mixed-schur-weyl-decomposition}
    (\mathbb{C}^d)^{\otimes n} \overset{\UC_d \times \SC_{n-1}}{\cong} \bigoplus_{\beta\vdash_d n-1} \,\bigoplus_{i\colon \beta_i>\beta_{i+1}} \cQ_{\alpha^{(i)}}^d \otimes \cP_{\beta},
\end{align}
where $\cQ_{\alpha^{(i)}}^d$ and $\cP_\beta$ are irreducible representations of $U(d)$ and $\SC_{m-1}$, respectively.
We denote their dimensions by $q_{\alpha^{(i)}}\equiv q_{\alpha^{(i)}}^d$ and $f_\beta$, respectively, and the action of $\pi\in \SC_{m-1}$ on $\cP_\beta$ by $p_\beta$.

\section{Distinguishing Symmetric Multipartite Werner States with PPT-BOTH Measurements}
With our preliminaries in place we are prepared to state our first theorem, which says that, when $n^2 \ll d$ PPT-BOTH measurements can only distinguish symmetric Werner states with a vanishingly small bias. 
\begin{theorem}\label{thm:symmetric-Werner-distinguishing}[Indistinguishability of Symmetric Werner States] Consider density matrices $\rho_b$ on $(\mathbb{C}^d)^{\otimes n}$ such that $[\rho_b,U^{\otimes n}] = [\rho_b,P_d(\pi)]=0$ for $b \in \{0,1\}$. If $n \geq 2$ and $3n(n-1) < 2d$, any PPT-BOTH measurement $\{M,\Id-M\}$ used to distinguish $\rho_0$ from $\rho_1$ will achieve bias at most
\begin{align}
    \abs{\tr{M(\rho_0-\rho_1)}} \leq \frac{3}{2}  \frac{n(n-1)}{d}.
\end{align}
\end{theorem}
\begin{proof} The high-level idea of this proof, as outlined in Sec.~\ref{sec:overview-main-results}, is to rewrite the difference of the two states using a telescoping sum. This will allow us to boil the entire problem down to bounding the trace norm of a partially transposed operator which depends on the marginals of the symmetric Werner states. Due to their symmetry, these states admit a rather simple decomposition which, in turn, allows us to exactly compute the spectrum of the resultant operator. This exact spectrum calculation is non-trivial and somewhat lengthy, so we only sketch it in the main text. A detailed proof is provided in Appendix~\ref{app:proofs-of-main-results}. Finally, simple bounds from Sec.~\ref{sec:branching-rules-young-identities} allow us to conclude the desired result.

To begin, recall that we denote the maximally mixed state on $\mathbb{C}^d$ as $\tau_d = \Id/d.$ Now, for any density matrix $\rho$ on $(\mathbb{C}^d)^{\otimes n}$, define the marginal obtained by tracing out qudits $m+1$ to $n$ as 
\begin{align}
    \rho^{(m)} \coloneq \operatorname{tr}_{m+1,\ldots,n}(\rho).
\end{align}
Now, consider the difference $\rho - \tau_d^{\otimes n}$. We can rewrite this as a telescoping sum to obtain
\begin{align}
    \rho- \tau_d^{\otimes n} &= \rho - \rho^{(n-1)}\otimes \tau_d  + \rho^{(n-1)}\otimes \tau_d - \dotsm - \rho^{(2)} \otimes \tau_d^{\otimes (n-2)}+\rho^{(2)} \otimes \tau_d^{\otimes (n-2)} - \tau_d^{\otimes n}.
\end{align}
When $\rho$ is a Werner state, we may simplify further because $[\rho,U^{\otimes n}]$ for all $U \in \UC_d$ implies that $\rho^{(1)}=\tau_d$. This, along with the identification $\rho = \rho^{(n)}$, allows us to write the difference in a compact form as 
\begin{align}
    \rho -\tau_d^{\otimes n} = \sum_{m=2}^n (\rho^{(m)} - \rho^{(m-1)}\otimes \tau_d) \otimes \tau_d^{\otimes (n-m)}.
\end{align}
This holds for any (not necessarily symmetric) Werner state. Now, let $\rho_0$ and $\rho_1$ be two distinct symmetric Werner states. Using $\rho_0 - \rho_1 = \rho_0 - \tau_d^{\otimes n} - (\rho_1 - \tau_d^{\otimes n})$, we may write
\begin{align}
    \rho_0 - \rho_1 = \sum_{b=0}^1 (-1)^b \sum_{m=2}^n (\rho_b^{(m)} - \rho_b^{(m-1)}\otimes \tau_d)\otimes \tau_d^{\otimes (n-m)}.
\end{align}
Having rewritten our difference of states in this fashion, we are ready to upper-bound the magnitude of the bias in terms the trace norm of a partially transposed operator. Using triangle inequality, we may begin upper-bounding the bias as
\begin{align}
    \abs{\tr{M(\rho_0-\rho_1)}} &\leq \sum_{b=0}^1 \sum_{m=2}^n \abs{\tr{M ((\rho_b^{(m)} - \rho_b^{(m-1)}\otimes \tau_d)\otimes \tau_d^{\otimes (n-m)})}}.
\end{align}
Now, $\rho_0$ and $\rho_1$ are completely arbitrary symmetric Werner states, so any uniform upper bound will apply to both. As, we will suppress the dependence on $b$ carry the resultant factor of two through the subsequent steps. Continuing, we have
\begin{align}
    \abs{\tr{M(\rho_0-\rho_1)}} &\leq 2 \sum_{m=2}^n \abs{\tr{M ((\rho^{(m)} - \rho^{(m-1)}\otimes \tau_d)\otimes \tau_d^{\otimes (n-m)})}} \\
    &= 2 \sum_{m=2}^n \abs{\tr{M^{\Gamma_m} ((\rho^{(m)} - \rho^{(m-1)}\otimes \tau_d)\otimes \tau_d^{\otimes (n-m)})^{\Gamma_m}}} \\
    &\leq \sum_{m=2}^n \|(\rho^{(m)} - \rho^{(m-1)}\otimes \tau_d)\otimes \tau_d^{\otimes (n-m)})^{\Gamma_m} \|_1 \quad \text{(by Lemma~\ref{lem:overlap-to-one-norm-bound})}\\
    &= \sum_{m=2}^n \|((\rho^{(m)})^{\Gamma_m} - \rho^{(m-1)}\otimes \tau_d)\otimes \tau_d^{\otimes (n-m)} \|_1\\
    &= \sum_{m=2}^n \|(\rho^{(m)})^{\Gamma_m} - \rho^{(m-1)}\otimes \tau_d \|_1,
\end{align}
where in the penultimate equality we have used that $\Gamma_m$ acts non-trivially only on the $m$-th tensor site and in the last inequality we used the multiplicativity of the trace norm and the fact that $\|\tau_d\|_1=1$. We also note that the application of Lemma~\ref{lem:overlap-to-one-norm-bound} is valid only for PPT-BOTH effects (i.e. those satisfying $0 \preceq M^{\Gamma_S} \preceq \Id$ for all $S \subseteq [n]$. In our case, we actually use the weaker condition that the effects remain positive under the partial transpose of any single tensor site. 

To simplify further, we must use the symmetry of $\rho$. Because $\rho$ is a symmetric Werner state, every marginal
$\rho^{(m)}$ is a symmetric Werner state on the smaller space. That is, the marginals remain invariant under the collective $\UC_d$ action and
the $\SC_m$ action. 
Recalling the Schur--Weyl decomposition as discussed in Sec.~\ref{sec:schur-weyl},
\begin{align}
      (\mathbb{C}^d)^{\otimes m}
    \cong
    \bigoplus_{\lambda\vdash_d m}
    \mathcal{Q}_\lambda^d\otimes\mathcal{P}_\lambda,
    \label{eq:schur-weyl}
\end{align}
and the orthogonal projections $\Pi_\lambda^{(m)}$ onto the isotpyical components
$\mathcal{Q}_\lambda^d\otimes\mathcal{P}_\lambda$, we set
\begin{align}
    q_\lambda\coloneqq\dim\mathcal{Q}_\lambda^d,
    \qquad
    f_\lambda\coloneqq\dim\mathcal{P}_\lambda,
    \qquad
    \omega_\lambda^{(m)}
    \coloneqq
    \frac{\Pi_\lambda^{(m)}}{q_\lambda f_\lambda}.
\end{align}
With this notation in place, all marginals of symmetric Werner states admit a convex combination of the form (see also Def.~\ref{def:Werner-states})
\begin{align}
    \rho^{(m)} = \sum_{\lambda \vdash_d m} p_m(\lambda) \omega_{\lambda}^{(m)}, \quad \text{where} \quad p_m(\lambda) \geq 0 \quad \text{and} \quad \sum_{\lambda \vdash_d m} p_m(\lambda) =1.
\end{align}
To use this decomposition, note that $\rho^{(m-1)} = \operatorname{tr}_m[\rho^{(m)}]$ and that both the partial trace and partial transpose are linear. Using these facts along with the convexity of the trace norm, we obtain
\begin{align}
    \abs{\tr{M(\rho_0-\rho_1)}} 
    &\leq \sum_{m=2}^n \sum_{\lambda \vdash_d m} p_m(\lambda) \|(\omega_{\lambda}^{(m)})^{\Gamma_m} - \operatorname{tr}_m[\omega_{\lambda}^{(m)}] \otimes \tau_d\|_1.
\end{align}
We have now reduced the entire problem to computing or upper bounding this trace norm. In Proposition~\ref{prop:exact-spectrum} and Corollary~\ref{cor:exact-trace-norm-bound}, we will show that 
\begin{align}
    \|X_{\lambda}\|_1 \leq \frac{3(m-1)}{d}, \quad \text{where} \quad X_{\lambda} \coloneq (\omega_{\lambda}^{(m)})^{\Gamma_m} - \operatorname{tr}_m[\omega_{\lambda}^{(m)}] \otimes \tau_d,
\end{align}
which we can substitute back into our bias bound to obtain
\begin{align}
    \abs{\tr{M(\rho_0-\rho_1)}} 
    &\leq \sum_{m=2}^n \sum_{\lambda \vdash_d m} p_m(\lambda) \frac{3 (m-1)}{d} = \frac{3}{2}\frac{n(n-1)}{d},
\end{align}
where we have used the fact that, for fixed $m$, $\sum_{\lambda \vdash_d m} p_m(\lambda) =1$. This completes the proof.
\end{proof}
We must now prove the essential trace norm bound. In the case of symmetric Werner states, we can exactly compute the trace norm that arises. This computation leverages facts from both the standard and mixed Schur-Weyl duality reviewed in Section~\ref{sec:schur-weyl} and Section.~\ref{sec:mixed-Schur-Weyl}, respectively. Recall from this section that $\lambda\vdash_d m$ be a Young diagram with $m$ boxes and at most $d$ rows, and consider the associated isotypical projector onto the $\lambda$-isotypical component,
\begin{align}
    \Pi_\lambda = \frac{f_\lambda}{m!}\sum_{\pi\in S_m} \chi_\lambda(\pi) P_d(\pi) \in \End((\mathbb{C}^d)^{\otimes m}),
\end{align}
where $P_d(\pi)$ denotes the tensor permutation representation of $S_m$ on $(\mathbb{C}^d)^{\otimes m}$. Having recalled this notation, we may write the operator whose trace norm we wish to bound as
\begin{align}
    X_\lambda \coloneqq \frac{1}{q_\lambda f_\lambda} \left(\Pi_\lambda^{\Gamma_m} - \trc_m \Pi_\lambda \otimes \tau_d\right).
\end{align}
Crucially, it is invariant under the action of $U(d) \times \SC_{m-1}$ defined in Sec.~\ref{sec:mixed-Schur-Weyl}, and hence 
\begin{align}
    X_\lambda \cong \bigoplus_{\beta\vdash_d m-1} \,\bigoplus_{i\colon \beta_i>\beta_{i+1}} x_{\beta,i} \Id_{\cQ_{\alpha^{(i)}}^d}\otimes \Id_{\cP_\beta}
    \label{eq:X_lambda-decomposition}
\end{align}
for some coefficients $x_{\beta,i}\in\mathbb{R}$.
These coefficients are the eigenvalues of $X_\lambda$ with multiplicities $q_{\alpha^{(i)}} f_\beta$.

In the following result, we write $\beta \prec \lambda$ for partitions $\beta\vdash m-1$ and $\lambda\vdash m$ if $\beta$ is obtained from $\lambda$ by removing a single box, and we denote by $c(\lambda\setminus\beta) = j-i$ the content of the box at position $(i,j)$ in $\lambda$ that was removed to obtain $\beta$. Given this set-up, we may now state and sketch the proof of the essential proposition. A detailed proof is given in Appendix~\ref{app:proofs-of-main-results}.
\begin{proposition} \label{prop:exact-spectrum}
    For $\lambda\vdash_d m$, the coefficients $x_{\beta,i}$ in \eqref{eq:X_lambda-decomposition} are determined as follows:
    \begin{enumerate}[label=(\normalfont\roman*)]
        \item If $\beta \not\prec \lambda$, then $x_{\beta,i} = 0$.
        \item If $\beta\prec\lambda$ and $\alpha = \beta-\varepsilon_i$ is not a Young diagram (i.e., $\alpha_{d} = -1$), then
        \begin{align}
            x_{\beta,i} = \frac{1}{m q_\lambda f_\beta} - \frac{1}{d f_\lambda q_\beta}.
        \end{align}
        \item If $\beta\prec\lambda$ and $\alpha = \beta-\varepsilon_i$ is a Young diagram (i.e., $\alpha_d\geq 0$), then
        \begin{align}
            x_{\beta,i} = \frac{1}{m q_\lambda f_\beta} \left( 1 + \frac{(m-1)q_\beta f_\alpha}{f_\beta\, q_\alpha (c(\lambda\setminus \beta)-c(\beta\setminus\alpha))} \right) - \frac{1}{df_\lambda q_\beta}.
        \end{align}
    \end{enumerate}
\end{proposition}
\begin{proof}[Proof Sketch]
    The proof strategy uses the following simple observation: Let $\cH\cong \bigoplus_i \cH_i$ be a direct sum decomposition of a Hilbert space $\cH$ with associated orthogonal projections $P_i$ onto $\cH_i$.
    If an operator is of the form $A \cong \bigoplus_i a_i \Id_{\cH_i} = \sum_i a_i P_i$, then we have $\trc(P_iA) = a_i\trc P_i = a_i \dim\cH_i$, 
    which allows us to determine the coefficients $a_i$.

    Here, the direct sum decomposition is given by the isotypical decomposition \eqref{eq:mixed-schur-weyl-decomposition} of $(\mathbb{C}^d)^{\otimes m}$ with projectors $(\Pi^{(m-1)}_\beta\otimes \Id_m) \Theta^d_{\alpha^{(i)}}$, where $\Pi^{(m-1)}_\beta$ is defined in Sec.~\ref{sec:schur-weyl}, and the $\Theta^d_{\alpha^{(i)}}$ are the isotypical projectors for the representation $U\mapsto U^{\otimes m-1}\otimes\overline{U}$ of $\UC_d$ discussed in Sec.~\ref{sec:mixed-Schur-Weyl}.
    The operator $A$ is of the form $A = c_1 (\Pi^{(m)}_\lambda)^{\Gamma_m} + c_2 \trc_m (\Pi^{(m)}_\lambda) \otimes \Id_m$.
    The $\lambda$-isotypical projector $\Pi^{(m)}_\lambda$ can be expressed using the character formula
    \begin{align} 
        \Pi^{(m)}_\lambda = \frac{f_\lambda}{m!} \sum_{\pi\in \SC_m} \chi_\lambda(\pi) P_d(\pi).
        \label{eq:isotypical-projector-character-formula}
    \end{align}
    
    To compute the partial transpose and the partial trace of $\Pi^{(m)}_\lambda$, it is useful to isolate the last system by splitting up the summation in \eqref{eq:isotypical-projector-character-formula} into two parts: one over the subgroup $G\cong S_{m-1}$ of those permutations fixing $m$, and one over the left cosets of $G$ of the form $(k,m)G$ for $k=1,\dots,m-1$. 
    For the summation over $G$ we make use of the branching rule \eqref{eq:branching-rule} for the symmetric group, whereas for the summation over the left cosets of $G$ we use the well-known fact that $(P_d(k,m))^{\Gamma_m}$ is proportional to a maximally entangled state on systems $k$ and $m$.
    These observations allow us to evaluate the traces of $(\Pi^{(m)}_\lambda)^{\Gamma_m}$ and $\trc_m (\Pi^{(m)}_\lambda)$ against the projectors $(\Pi^{(m-1)}_\beta\otimes \Id_m) \Theta^d_{\alpha^{(i)}}$, and thus determine the coefficients $x_{\beta,i}$ in \eqref{eq:X_lambda-decomposition}.
\end{proof}

Starting from the exact spectrum above, we can compute the exact trace norm and then use basic branching/dimension identities to recover the upper bound we obtain below.
\begin{corollary} \label{cor:exact-trace-norm-bound}In the setting of the above proposition, let $s_{\beta} \coloneq d q_{\beta} - \sum_{\alpha \prec \beta} q_{\alpha}$. We can then show
\begin{align}
    \|X_{\lambda}\|_1 &= \sum_{\beta \prec \lambda} \sum_{\alpha \prec \beta} \frac{q_{\alpha} f_{\beta}}{q_{\beta} f_{\lambda}} \abs{\frac{1}{c(\lambda \setminus \beta) - c(\beta \setminus \alpha)} - \frac{1}{d}} + \frac{1}{md q_{\lambda} } \sum_{\beta \prec \lambda} s_{\beta} \abs{c(\lambda \setminus \beta)}.
\end{align}
If we further assume $d > 3n(n-1)$ (which is the regime in which Theorem~\ref{thm:symmetric-Werner-distinguishing} is non-trivial), we can show that 
\begin{align}
    \|X_{\lambda}\|_1 \leq (1+\frac{1}{d}) \cdot \frac{m-1}{d-m+2} + \frac{m-1}{d-m+1} \leq \frac{3(m-1)}{d}.
\end{align}
\end{corollary}
\begin{proof}
    We can start by simplifying the form of the eigenvalues using the one-box identities. Recall that when $\beta \prec \lambda$, we have
    \begin{align}
        \frac{q_{\lambda}/f_{\lambda}}{q_{\beta} / f_{\beta}} = \frac{d+c(\lambda \setminus \beta)}{m} \Longleftrightarrow m q_{\lambda} f_{\beta} = (d+ c(\lambda \setminus \beta))q_{\beta} f_{\lambda}.
    \end{align}
    Similarly, for $\alpha \prec \beta$, we have
    \begin{align}
        \frac{q_{\beta} / f_{\beta}}{q_{\alpha}/ f_{\alpha}} = \frac{d+c(\beta \setminus \alpha)}{m-1} \Longleftrightarrow (m-1) \frac{q_{\beta} f_{\alpha}}{f_{\beta}q_{\alpha}} = d + c(\beta \setminus \alpha).
    \end{align}
    Using these two identities, we can re-write the eigenvalues for the polynomial sector as 
    \begin{align}
        x_{\beta, \alpha}^{\rm poly} = \frac{1}{q_{\beta} f_{\lambda}} \left( \frac{1}{c(\lambda \setminus \beta) - c(\beta \setminus \alpha)}-\frac{1}{d}\right), \quad \text{with multiplicities} \quad q_{\alpha} f_{\beta}.
    \end{align}
    We can do a similar thing for the non-polynomial eigenvalues. Using $m q_{\lambda} f_{\beta} =(d+c(\lambda \setminus \beta)) q_{\beta} f_{\lambda}$, we can write the non-polynomial eigenvalues as
    \begin{align}
        x_{\beta}^{\rm non} &= \frac{1}{m q_{\lambda} f_{\beta}} - \frac{1}{d f_{\lambda} q_{\beta}} = \frac{1}{(d+c(\lambda \setminus \beta))q_{\beta}f_{\lambda}} - \frac{1}{d q_{\beta} f_{\lambda}} = -\frac{c(\lambda \setminus \beta)}{d(d+c(\lambda \setminus \beta)) q_{\beta} f_{\lambda}},
    \end{align}
    which we can simplify further by again using $m q_{\lambda} f_{\beta} = (d+c(\lambda \setminus \beta))q_{\beta}f_{\lambda}$. This yields the simplified eigenvalues on the non-polynomial sector
    \begin{align}
         x_{\beta}^{\rm non}  &= - \frac{c(\lambda \setminus \beta)}{m d q_{\lambda} f_{\beta}}, \quad \text{with multiplicity} \quad f_{\beta} \left(d q_{\beta} - \sum_{\alpha \prec \beta} q_{\alpha}\right) \eqcolon f_{\beta}s_{\beta}.
    \end{align}
    Putting these together, we obtain the exact formula for the trace norm
    \begin{align}
         \|X_{\lambda}\|_1 &= \sum_{\beta \prec \lambda} \sum_{\alpha \prec \beta} \frac{q_{\alpha} f_{\beta}}{q_{\beta} f_{\lambda}} \abs{\frac{1}{c(\lambda \setminus \beta) - c(\beta \setminus \alpha)} - \frac{1}{d}} + \frac{1}{md q_{\lambda}} \sum_{\beta \prec \lambda} s_{\beta} \abs{c(\lambda \setminus \beta)}.
    \end{align}
    Now, let us derive the upper bound that makes the asymptotic scaling readily apparent. For starters, note that $c(\lambda \setminus \beta) - c(\beta \setminus \alpha)$ is a non-zero integer, so 
    \begin{align}
        \abs{\frac{1}{c(\lambda \setminus \beta)-c(\beta \setminus \alpha)}} \leq 1.
    \end{align}
    Applying this bound, we may simplify the polynomial terms further as
    \begin{align}
         \sum_{\beta \prec \lambda} \sum_{\alpha \prec \beta} \frac{q_{\alpha f_{\beta}}}{q_{\beta} f_{\lambda}} \abs{\frac{1}{c(\lambda \setminus \beta) - c(\beta \setminus \alpha)} - \frac{1}{d}} \leq \left(1+\frac{1}{d}\right) \sum_{\beta \prec \lambda} \frac{f_{\beta}}{f_{\lambda}} \sum_{\alpha \prec \beta} \frac{q_{\alpha}}{q_{\beta}} \leq \left(1+\frac{1}{d}\right) \frac{m-1}{d-m+2},
    \end{align}
    where we have used the fact that $\sum_{\beta \prec \lambda} f_{\beta} = f_{\lambda}$ along with the second part of Lemma~\ref{lem:one-box-dimension-bounds}. The non-polynomial terms also simplify very nicely. First, note
    \begin{align}
        s_{\beta}=dq_{\beta}-\sum_{\alpha \prec \beta} q_{\alpha} \leq dq_{\beta},
    \end{align}
    and $\abs{c(\lambda \setminus \beta)} \leq m-1$. Thus, we may write
    \begin{align}
        \frac{1}{md q_{\lambda}} \sum_{\beta \prec \lambda} s_{\beta} \abs{c(\lambda \setminus \beta)} &\leq \frac{m-1}{m q_{\lambda}} \sum_{\beta \prec \lambda} q_{\beta} \leq \frac{m-1}{d-m+1},
    \end{align}
    where in the last inequality we have applied the first part of Lemma~\ref{lem:one-box-dimension-bounds}. Putting all of this together, we obtain
    \begin{align}
    \|X_{\lambda}\|_1 \leq \left(1+\frac{1}{d}\right) \cdot \frac{m-1}{d-m+2} + \frac{m-1}{d-m+1} \leq \frac{3(m-1)}{d},
\end{align}
where the last inequality holds when $m\leq n$ and $d > 3n(n-1).$
\end{proof}
\subsection{Applications to Testing Unitarily Invariant Properties}
The preceding theorem also gives sample complexity lower bounds for
quantum property testing~\cite{montanaro2016Survey}. For this application, we interpret the factors
of $(\mathbb C^d)^{\otimes n}$ as \emph{copy registers} containing $n$
independent copies of an unknown state, rather than as spatially
separated parties. An adaptive single-copy protocol measures one
copy at a time, with the $k$-th measurement potentially depending on all $k-1$ previous outcomes. The resultant two-outcome measurement has separable effects
across the copy registers and therefore satisfies PPT-BOTH. Theorem~\ref{thm:symmetric-Werner-distinguishing} thus implies the following corollary.

\begin{corollary}[Testing Unitarily Invariant Properties] \label{cor:testing-unitarily-invariant-properties} Any PPT-BOTH tester $\{M,\Id-M\}$ on $(\mathbb{C}^d)^{\otimes n}$ for a non-trivial unitarily invariant property on $\PC \subseteq \BC(\mathbb{C}^d)$ requires $n \geq \Omega(\sqrt{d})$ samples to achieve constant bias.
\end{corollary}
\begin{proof}
Let $\mathcal P \subseteq \BC(\mathbb{C}^d)$ be a unitarily invariant property, and fix $\eps>0$. A tester must accept every $\rho\in\mathcal P$ and reject every state satisfying
\begin{align}
 \inf_{\sigma\in\mathcal P}\frac12\|\rho-\sigma\|_1
 \ge\eps,
\end{align}
with probability at least $2/3$. Choose any state $\rho_{\rm Y}\in\mathcal P$ and any state $\rho_{\rm N}$ at trace distance at least $\eps$ from $\mathcal P$. Drawing $U$ according to Haar measure on $\UC_d$ gives the ensembles
\begin{align}
 \mathcal E_{\mathrm Y}
 &=\{U\rho_{\mathrm Y}U^\dagger\}_U \quad \text{and} \quad
 \mathcal E_{\mathrm N}
 =\{U\rho_{\mathrm N}U^\dagger\}_U.
\end{align}
Unitary invariance of the property and the trace norm ensures that
every member remains in its respective promise class. For $n$ copies,
their averaged states are
\begin{align}
 \omega_b^{(n)}
 \coloneq
 \int_{\UC_d}
 \bigl(U\rho_bU^\dagger\bigr)^{\otimes n}\,dU,
 \qquad b\in\{\mathrm Y,\mathrm N\}.
\end{align}
Here the unitary is drawn once, and the resulting unknown state is
supplied in $n$ copies. Haar averaging makes $\Omega_b^{(n)}$
invariant under collective unitaries (by left-invariance of the Haar measure), while each tensor power is invariant under permutations of the copies. Thus both averaged states are symmetric multipartite Werner states. Now, let $M$ denote the tester's acceptance effect. Since the tester
succeeds on every promised input, averaging its acceptance
probabilities gives
\begin{align}
 \tr{M\omega_{\rm Y}^{(n)}}
 &\ge \frac{2}{3}, \quad \text{and} \quad
 \tr{M\omega_{\mathrm N}^{(n)}}
 \le \frac{1}{3}.
\end{align}
Consequently, the tester distinguishes the two averaged states with
bias at least $1/3$. If its measurement satisfies PPT-BOTH,
Theorem~\ref{thm:symmetric-Werner-distinguishing} implies
\begin{align}
\frac{1}{3}
 \le
 \operatorname{tr}\!\left[
 M\bigl(\omega_{\rm Y}^{(n)}-\omega_{\rm N}^{(n)}\bigr)
 \right]
 \le
 \frac{3}{2}\frac{n(n-1)}{d}
\end{align}
whenever $n \geq 2$ and $3n(n-1)<2d$. This forces
$n=\Omega(\sqrt d)$ and we are done.
\end{proof}
We therefore obtain an $\Omega(\sqrt d)$ copy lower bound for
PPT-BOTH, and thus adaptive single-copy, testers for any non-trivial unitarily invariant property. We illustrate the power of this result with two concrete examples: purity and low-rank testing.

\medskip
\noindent\textbf{Purity testing.}
Consider testing whether an unknown state on $\mathbb C^d$ is pure
or $\epsilon$-far in trace distance from every pure state.
To prove a sample complexity lower bound on this task for any fixed $0<\epsilon\leq 1/2$ and $d\geq2$, it suffices to consider distinguishing Haar-random pure states and the maximally mixed state $\tau_d$. Their averaged $n$-copy states are (cf. Lemma~\ref{lem:haar-twirl-symmetric-subspace})
\begin{align}
 \omega_{\rm Y}^{(n)} &= \underset{U \in \UC_d}{\mathbb{E}}\left[U^{\otimes n} \ketbra{\phi}{\phi}^{\otimes n}(U^\dagger)^{\otimes n}\right]
 = \frac{\Pi_n}{d[n]} \quad \text{and} \quad
 \omega_{\rm N}^{(n)}=\tau_d^{\otimes n}
 =\frac{\Id}{d^n},
\end{align}
where $\Pi_n$ projects onto the symmetric subspace $\vee^n \mathbb{C}^d$. Note $\tau_d$ is at trace distance $1-1/d\geq1/2$ from every pure state, thus a PPT-BOTH tester, should one exist, could be used to distinguish these ensemble-averaged states. Both averaged states are symmetric Werner states, so our
indistinguishability bound along with the assumption that a tester exists, yields
\begin{align}
 \Omega(1) \abs{\tr{M\bigl(\omega_{\mathrm Y}^{(n)}-\omega_{\mathrm N}^{(n)}\bigr)}}
 \leq O\left(\frac{n^2}{d}\right) \implies n \geq \Omega(\sqrt{d}).
\end{align}
This recovers the bound from Ref.~\cite{akresh2026Optimal}, in which the authors directly prove this specific result using only linear algebra and a few facts from the church of the symmetric subspace~\cite{harrow2013Church}. This is a strengthening of the original bound from Ref~\cite{chen2022Exponential} which uses the learning tree formalism to prove an adaptive single-copy lower bound on this task. While it is nice to recover a known result as a special case, let us demonstrate that our result is also capable of giving new bounds.

\medskip
\noindent\textbf{Low-rank testing.}
Consider testing whether an unknown state on $\mathbb C^d$ has
rank at most $r$ or is $\epsilon$-far in trace distance from
every such state~\cite{childs2007Weak,odonnell2015Quantum}. For $d\geq2r$ and any fixed
$0<\epsilon\leq1/2$, it suffices to distinguish states that are
maximally mixed on Haar-random subspaces of dimensions $r$ and
$2r$, respectively. Writing $P_s$ for a fixed rank-$s$
orthogonal projector, their averaged $n$-copy states are
\begin{align}
 \omega_{\rm Y}^{(n)}
 &= \underset{U\in\UC_d}{\mathbb E}
 \left[\left(\frac{UP_rU^\dagger}{r}\right)^{\otimes n}\right] \quad \text{and} \quad
 \omega_{\rm N}^{(n)}
 = \underset{U\in\UC_d}{\mathbb E}
 \left[\left(\frac{UP_{2r}U^\dagger}{2r}\right)^{\otimes n}\right].
\end{align}
Every state in the first ensemble has rank $r$, while every
state in the second is at trace distance $1/2$ from the set
of states of rank at most $r$. To see this, a state maximally
mixed on a $2r$-dimensional subspace assigns weight at most
$1/2$ to the support of any state of rank at most $r$.
This gives a trace-distance lower bound of $1/2$, attained
by the normalized projector onto any $r$-dimensional
subspace of its support. Both ensemble-averaged states
are symmetric Werner states. Consequently, in the regime of
Theorem~\ref{thm:symmetric-Werner-distinguishing}.
As in the general reduction above, this implies an
$\Omega(\sqrt d)$ sample complexity lower bound.

For comparison, Theorem~40 of Ref.~\cite{ye2025Exponential},
specialized to single-copy measurements and spectra uniform
on $r$ and $2r$ entries, implies a lower bound of
$\Omega(\sqrt d/[r\sqrt{\log(2r)}])$ for the same
distinguishing task. Our result removes this rank-dependent
denominator and extends the lower bound from adaptive
single-copy protocols to all PPT-BOTH measurements.
For fixed $r$, the dimension dependence agrees; the improvement
in rank dependence becomes relevant when $r$ grows with $d$.

These examples illustrate a general consequence of our result. In their review on quantum property testing, Montanaro and de Wolf show that any tester, with unrestricted measurement access, for a unitarily invariant property can be replaced by weak Schur sampling followed by classical post-processing, without worsening its completeness, soundness, or copy complexity~\cite[Lemma~20]{montanaro2016Survey}. Thus, for any such property, a dimension-independent weak Schur sampling algorithm, when taken together with Corollary~\ref{cor:testing-unitarily-invariant-properties}, yields a sample complexity separation: collective measurements suffice with $O(1)$ copies, while PPT-BOTH measurements, including adaptive single-copy protocols, require $\Omega(\sqrt d)$ copies. Here the testing accuracy, success probability, and any other property parameters are held fixed as $d$ grows.

\section{Distinguishing Multipartite Werner States with PPT-BOTH Measurements}
The proof of Theorem~\ref{thm:symmetric-Werner-distinguishing} relied heavily on the fact that symmetric Werner states have both collective unitary and permutation symmetry. Removing permutation symmetry makes directly computing the spectrum of the resultant operator highly non-trivial. Fortunately, our main result shows that there exists a method of upper-bounding the required trace norm that still yields the same bound on the bias.

\begin{theorem}[PPT-BOTH Indistinguishability of Multipartite Werner States] \label{thm:multipartite-werner-indistinguishability} 
    Consider density matrices $\rho_b$ on $(\mathbb{C}^d)^{\otimes n}$ such that $[\rho_b,U^{\otimes n}] = 0$ for all $U\in \UC_d$ and $b \in \{0,1\}$. If $n \geq 2$ and $3n(n-1) < 2d$, any PPT-BOTH measurement $\{M,\Id-M\}$ used to distinguish $\rho_0$ from $\rho_1$ will achieve bias at most
\begin{align}
    \abs{\tr{M(\rho_0-\rho_1)}} \leq \frac{3}{2}  \frac{n(n-1)}{d}.
\end{align}
\end{theorem}
\begin{proof}
    The proof of this result begins along identical lines to the proof of Thm.~\ref{thm:symmetric-Werner-distinguishing}. In fact, until the inequality
\begin{align}
    \abs{\tr{M(\rho_0-\rho_1)}} 
    &\leq \sum_{m=2}^n \|(\rho^{(m)})^{\Gamma_m} - \rho^{(m-1)}\otimes \tau_d \|_1,
\end{align}
we did not use the additional permutation symmetry that symmetric Werner states possess anywhere in the proof.\footnote{Note that $\rho$ represents an arbitrary Werner state here. We suppress the subscripts from the theorem statement to ease notation. See the proof of Theorem~\ref{thm:symmetric-Werner-distinguishing} for proper justification.} 
The first change in the proof is adjusting the decomposition of the $\rho^{(m)}$, as the input states and their marginals are still Werner states, but not necessarily symmetric. 
Thus, using Prop.~\ref{prop:Werner-marginals}, we may write
\begin{align}
 \rho^{(m)} =
 \bigoplus_{\lambda\vdash_d m}
 p_\lambda\frac{\Id_{\QC_\lambda^d}}{q_\lambda}
 \otimes\sigma_\lambda^{(m)} \eqcolon \bigoplus_{\lambda\vdash_d m}
 p_\lambda \Omega_{\lambda,\sigma}.
 \label{eq:werner-state-marginal}
\end{align}
Using linearity of both the partial trace and partial transpose, along with the convexity of the trace norm, as in the proof of Thm.~\ref{thm:symmetric-Werner-distinguishing}, we obtain an upper bound on the bias of the form
\begin{align}
     \abs{\tr{M(\rho_0-\rho_1)}} 
    &\leq \sum_{m=2}^n \sum_{\lambda \vdash_d m} p_{\lambda} \|X_{\lambda,\sigma}\|_1, \quad \text{with} \quad X_{\lambda,\sigma} \coloneq \Omega_{\lambda,\sigma}^{\Gamma_m} -  \trc_m \Omega_{\lambda,\sigma}\otimes \tau_d,
\end{align}
where we have suppressed the $m$ dependence to avoid over-cluttering the notation. Unlike the symmetric case, we will not exactly compute the spectrum of this operator. Rather, we will use the symmetries it satisfies to derive a direct sum decomposition of the operator into two distinct sectors. 

To be concrete, it will be helpful to make the decomposition of the $(\UC_d \times \SC_{m-1})$-representation space $(\mathbb{C}^d)^{\otimes m}$ in \eqref{eq:mixed-schur-weyl-decomposition} more specific, for which we will follow the analysis and notation of \cite{fei2023Efficient}, which is in turn based on a series of works on the algebra of partially transposed permutation operators \cite{studzinski2013Commutant,mozrzymas2014Structure} and port-based teleportation \cite{studzinski2017port,mozrzymas2018optimal,christandl2021Asymptotic,leditzky2022Optimality}. 

In \cite{fei2023Efficient}, the authors show that the decomposition \eqref{eq:mixed-schur-weyl-decomposition} of $(\mathbb{C}^d)^{\otimes m}$ can be rewritten as
\begin{align}
    (\mathbb{C}^d)^{\otimes m} \cong \cH_M \oplus \cH_S \label{eq:S-M-decomposition} \qquad
    \text{with}\qquad 
    \cH_M &= \bigoplus_{\alpha\vdash_d m-2} \cQ_{\alpha}^d \otimes \left( \bigoplus_{\beta\succ\alpha,\,\beta\vdash_{d} m-1} \cP_\beta \right)\\
    \cH_S &= \bigoplus_{\beta\vdash_{d-1} m-1} \cQ_{\beta^{-}}^d \otimes \cP_{\beta^0}, \label{eq:S-decomposition}
\end{align}
and for a partition $\beta\vdash_{d-1} m-1$ we set $\beta^- = (\beta_1,\dots,\beta_{d-1},-1)$ and $\beta^0 = (\beta_1,\dots,\beta_{d-1},0)$. Slightly overloading notation, we also denote by $S$ and $M$ the orthogonal projectors onto $\cH_S$ and $\cH_M$, respectively.
We write $S_\beta$ for the projector onto the summand $\cQ_{\beta^-}^d \otimes \cP_{\beta^0}$ in $\cH_S$, so that we have \begin{align} 
    S = \sum_{\beta\vdash_{d-1} m-1}S_\beta.
\end{align}

The decomposition \eqref{eq:S-M-decomposition} splits up the representation space $(\mathbb{C}^d)^{\otimes m}$ into a `polynomial' part $\cH_M$ comprising $\UC_d$-irreps associated with non-negative weights, and a `non-polynomial' part $\cH_S$ comprising those with a negative component.
The following claim further specifies the structure of the operator $X_{\lambda,\sigma}$ with respect to this decomposition.

\begin{claim}[Sector decomposition of $X_{\lambda,\sigma}$] \label{claim:sector-decomposition}
    With the notation above,
\begin{equation}
\label{eq:general-sector-split-sketch}
 X_{\lambda,\sigma}
 \cong MX_{\lambda,\sigma}M
 \oplus\bigoplus_{\beta\prec\lambda}
 S_\beta X_{\lambda,\sigma}S_\beta.
\end{equation}
The trace norm of $X_{\lambda,\sigma}$ decomposes accordingly as
\begin{align}
 \|X_{\lambda,\sigma}\|_1
 =\|MX_{\lambda,\sigma}M\|_1
 +\sum_{\beta\prec\lambda}\|S_\beta X_{\lambda,\sigma}S_\beta\|_1.
\end{align}
\end{claim}

\begin{proof}[Proof of Claim~\ref{claim:sector-decomposition}]
    Denote the operator of interest as $X_{\lambda,\sigma} = \Omega_{\lambda,\sigma}^{\Gamma_m} - \frac{1}{d} \trc_m \Omega_{\lambda,\sigma}\otimes \Id_m \eqcolon \Omega^{\Gamma_m} - \frac{1}{d}D$.
    Both $\Omega^{\Gamma_m}$ and $D$ are $U^{\otimes m-1}\otimes \overline{U}$-invariant so that
    \begin{align}
        X_{\lambda,\sigma} = MX_{\lambda,\sigma}M + S X_{\lambda,\sigma} S,
    \end{align}
    and according to \eqref{eq:S-decomposition} we further have
    \begin{align}
        SX_{\lambda,\sigma}S \cong \bigoplus_{\beta\vdash_{d-1}m-1} \Id_{\cQ_{\beta^-}^d}\otimes X_{\beta^0},
    \end{align}
    where the $X_{\beta^0}$ are some operators supported on $\cP_{\beta^0}$.

    To prove the claim of the lemma, we need to show that $X_{\beta^0}$ vanishes unless $\beta\prec \lambda$.
    We do this individually for $\Omega^{\Gamma_m}$ and $D$.
    To this end, we observe that $\cP_{\beta^0}$ in \eqref{eq:S-decomposition} can be identified with the irrep $\cP_\beta$ of $\SC_{m-1}$.
    Thus, 
    \begin{align}
        \trc\left( \left(\Pi_{\beta}^{(m-1)}\otimes \Id_m\right)\Omega^{\Gamma_m} \right) &= \trc\left( \left(\Pi_{\beta}^{(m-1)}\otimes \Id_m\right)^{\Gamma_m}\Omega \right)\\
        &= \trc\left( \left(\Pi_{\beta}^{(m-1)}\otimes \Id_m\right)\Omega \right)\\
        &= \trc\left( \left(\Pi_{\beta}^{(m-1)}\otimes \Pi_{(1)}^{(1)}\right) \Pi_\lambda^{(m)}\Omega\Pi_\lambda^{(m)} \right)\\
        &\neq 0 \quad\text{iff}\quad \beta\prec\lambda,
    \end{align}
    where we used $\trc(X^{\Gamma_m}Y^{\Gamma_m}) = \trc(XY)$ in the first equality, $\Pi_{(1)}^{(1)}=\Id$ in the third equality, and the Pieri rule for the product of projectors $\left(\Pi_{\beta}^{(m-1)}\otimes \Pi_{(1)}^{(1)}\right) \Pi_\lambda^{(m)}$ in the fourth line.

    On the other hand, for the operator $D=\trc_m \Omega\otimes \Id_m$ we recall Lem.~\ref{lem:general-marginal}, which says that 
    \begin{align} 
        \trc_m \Omega = \sum_{\alpha\prec\lambda} \Pi_\alpha^{(m-1)} \trc_m \Omega\Pi_\alpha^{(m-1)}.
    \end{align}
    Then,
    \begin{align}
        \trc\left( \left(\Pi_{\beta}^{(m-1)}\otimes \Id_m\right) \left(\trc_m \Omega\otimes \Id_m\right) \right) &= d\cdot  \trc\left( \Pi_{\beta}^{(m-1)} \trc_m \Omega\right)\\
        &= d \cdot \sum_{\alpha\prec\lambda} \trc\left(\Pi_{\beta}^{(m-1)} \Pi_\alpha^{(m-1)} \trc_m \Omega\Pi_\alpha^{(m-1)}\right)\\
        &\neq 0  \quad\text{iff}\quad \beta\prec\lambda,
    \end{align}
    which proves the claim.
\end{proof}
With this decomposition justified, we are only left to bound the trace norm of the polynomial and non-polynomial sectors individually.\\

\noindent \textbf{Bounding the  Polynomial Sector.} Despite the polynomial sector potentially retaining coherences between different branches, unlike the symmetric Werner case, we can still obtain a strong bound on its contribution to the overall trace norm. To obtain this bound, we will make use of the following observation.
\begin{claim}\label{claim:partial-transpose-CP}
    Let $X_{AB} \succeq 0$ be an operator on $\HC_A \otimes \HC_B$. Then,
    \begin{align}
        - \mathrm{tr}_B\left[X_{AB}\right] \otimes \Id_B \preceq X_{AB}^{\Gamma_B} \preceq \mathrm{tr}_B\left[X_{AB}\right] \otimes \Id_B.
    \end{align}
    Equivalently, $\mathrm{tr}_B\left[X_{AB}\right] \otimes \Id_B \pm X_{AB}^{\Gamma_{B}}$ are both positive semidefinite.
\end{claim}
A simple special case of this result provides some intuition. Consider the case when our bipartite operator is the projector onto the maximally entangled state $X_{AB} = \Phi_{AB}$. In this case, we know that $\mathrm{tr}_B\left[X_{AB}\right] = \Id_A/d$ and $X_{AB}^{\Gamma_B} = \swap_{AB}/d$, and the lemma states that
\begin{align}
    -\frac{\Id_{AB}}{d} \preceq \frac{\swap_{AB}}{d} \preceq \frac{\Id_{AB}}{d},
\end{align}
which is equivalent to saying that the swap operator's eigenvalues lie in the interval $[-1,1]$, which is easy to check. The general statement for arbitrary bipartite semidefinite operators follows from an application of Choi's Theorem~\cite[Thm.~2.22]{watrous2018Theory}.

\begin{proof}[Proof of Claim \ref{claim:partial-transpose-CP}]
   Let $d=\dim\HC_B$, and let $\HC_{B'}$ be an isomorphic copy of $\HC_B$.
   Define linear maps $\EC_{\pm}\colon \BC(\HC_B) \rightarrow \BC(\HC_B)$ by $\EC_{\pm}(X) = \tr{X} \Id_B \pm X^T$. These maps were chosen to ensure
   \begin{align}
       (\operatorname{id}_A \otimes \EC_{\pm})(X_{AB}) &= \mathrm{tr}_B\left[X_{AB}\right] \otimes \Id_B \pm X_{AB}^{\Gamma_B},
   \end{align}
   where $\operatorname{id}_A$ denotes the identity map on $\BC(\HC_A)$.
   Recall that Choi's Theorem tells us a linear map is completely positive if and only if the map's Choi matrix is positive semidefinite (PSD). Thus, to show that the right-hand side of the above equation is PSD, it suffices to show that the Choi matrix of $\EC_{\pm}$ is PSD. Recall that the \textit{Choi matrix} for a linear map $T : \BC(\HC_B) \rightarrow \BC(\HC_B)$ is defined as
   \begin{align}
       J(T) = \sum_{i,j=1}^d \ketbra{i}{j}_{B'} \otimes T(\ketbra{i}{j}_B).
   \end{align}
   To apply this to $\EC_{\pm}$, we note that $
       \EC_{\pm}(\ketbra{i}{j}_B) = \delta_{ij} \Id_B \pm \ketbra{j}{i}_B.$ Substituting this into the definition of the Choi matrix we find
   \begin{align}
       J(\EC_{\pm}) &= \sum_{i,j=1}^d \ketbra{i}{j}_{B'} \otimes (\delta_{ij} \Id_B \pm \ketbra{j}{i}_B) = \Id_{B'B} \pm \swap_{B'B}.
   \end{align}
   Noting that $\Id_{B'B} + \swap_{B'B} = 2~\Pi_{\rm sym}^{d,2} \succeq 0$ and $\Id_{B'B} - \swap_{B'B} = 2~\Pi_{\rm asym}^{d,2} \succeq 0$, we have that both Choi matrices are PSD. By Choi's Theorem, both maps $\EC_{\pm}$ are completely positive. Since $X_{AB}\succeq0$, it follows that
   \begin{align}
       (\operatorname{id}_A \otimes \EC_{\pm})(X_{AB})
       = \mathrm{tr}_B[X_{AB}] \otimes \Id_B
       \pm X_{AB}^{\Gamma_B} \succeq 0,
   \end{align}
   as required.
\end{proof}
Now, to apply this to bounding the polynomial sector, recall that $\Omega \succeq 0$ and $D \coloneq \operatorname{tr}_m[\Omega] \otimes \Id_m$. Considering the bipartition between the first $m-1$ sites and the last, the claim implies
\begin{align}
    D \pm \Omega^{\Gamma_m} \succeq 0 \implies -D \preceq \Omega^{\Gamma_m} \preceq D.
\end{align}
Projecting onto the polynomial irreps maintains the chain of inequalities yielding
\begin{align}
    -MDM \preceq M \Omega^{\Gamma_m} M \preceq MDM.
\end{align}
Moreover, because $M \Omega^{\Gamma_m} M$ is Hermitian and $MDM \succeq 0$, we have that\footnote{To see \eqref{eq:trace-norm-inequality}, let $X,Y$ be Hermitian with $-Y\preceq X \preceq Y$.
This assumption implies that $Y\succeq 0$. Let $P$ be the projector onto the support of the positive semidefinite part of $X$, and set $P^\perp = \Id-P$. Then we have $|X|= X_+ + X_- = PXP + P^\perp(-X)P^\perp$, and thus $$\|X\|_1 = \trc|X| = \trc(PX) + \trc(P^\perp(-X)) \leq \trc(PY) + \trc(P^\perp Y) = \trc((P+P^\perp)Y)=\trc Y = \|Y\|_1,$$ where we used the assumptions $X\preceq Y$ and $-X\preceq Y$, and $Y\succeq 0$.}  
\begin{align} 
    \|M\Omega^{\Gamma_m}M\|_1 \leq \tr{MDM} = \|MDM\|_1.
    \label{eq:trace-norm-inequality}
\end{align}
With this in mind, we are in place to bound the polynomial contribution to the trace norm as 
\begin{align}
    \|M X_{\lambda,\sigma} M\|_1 &= \left\| M \left(\Omega^{\Gamma_m} - \frac{1}{d} D \right) M\right\|_1,\\
    &\leq \|M \Omega^{\Gamma_m} M \|_1 + \frac{1}{d}\|MDM \|_1, \quad &\text{(by triangle inequality)}\\
    &\leq \left(1+ \frac{1}{d}\right) \|MDM\|_1, &\text{(by eq.~\eqref{eq:trace-norm-inequality})}\\
    &= \left(1+ \frac{1}{d}\right) \tr{MDM}.
\end{align}
To simplify this trace, note that Eq.~\eqref{eq:marginal-Schur-block} yields
\begin{align}
    \operatorname{tr}_m[\Omega_{\lambda,\sigma}] &= \bigoplus_{\beta \prec \lambda} \frac{1}{q_{\beta}}\Id_{\QC_\beta^d} \otimes \sigma_{\beta},
    \intertext{and thus}
    D &= \operatorname{tr}_m[\Omega_{\lambda,\sigma}]\otimes \Id_m = \bigoplus_{\beta \prec \lambda} \frac{1}{q_{\beta}}\Id_{\QC_\beta^d \otimes \mathbb{C}^d} \otimes \sigma_{\beta},
\end{align}
where we have grouped the appended identity on the $m$-th copy of $\mathbb{C}^d$ with the identity on the Weyl module. Now, under the dual-Pieri decomposition, $M$ retains only the polynomial summands $\QC_{\alpha}^d$ with $\alpha \prec \beta$ in each $\beta$-branch. Thus, in the mixed Schur basis used in \eqref{eq:mixed-schur-weyl-decomposition}, we have
\begin{align}
    (MDM)\rvert_{\HC_M} = \bigoplus_{\beta \prec \lambda} \left[\left(\bigoplus_{\alpha \prec \beta} \Id_{\QC_{\alpha}^d}\right) \otimes \frac{\sigma_{\beta}}{q_{\beta}}\right].
\end{align}
Taking the trace of this operator, then, yields
\begin{align}
    \tr{MDM} &= \sum_{\beta \prec \lambda} \left(\sum_{\alpha \prec \beta} q_{\alpha}\right) \frac{\tr{\sigma_{\beta}}}{q_{\beta}},\\
    &= \sum_{\beta \prec \lambda} t_{\beta}\sum_{\alpha \prec \beta} \frac{q_{\alpha}}{q_{\beta}}, \quad &\text{(with $t_{\beta} \coloneq \tr{\sigma_{\beta}}$)},\\
    &\leq \sum_{\beta \prec \lambda} t_{\beta} \frac{m-1}{d-m+2}, \quad & \text{(by Lemma~\ref{lem:one-box-dimension-bounds})},\\
    &= \frac{m-1}{d-m+2},
\end{align}
where the last equality uses the fact that $\sum_{\beta \prec \lambda} t_{\beta} =1$. Plugging this into the trace norm we set out to bound, we obtain
\begin{align}
    \|MX_{\lambda, \sigma} M\|_1 \leq \left(1+\frac{1}{d}\right) \frac{m-1}{d-m+2}.
\end{align}

\noindent \textbf{Bounding the  Non-polynomial Sector.} 
To bound the non-polynomial contribution to the trace norm, we begin by expanding the 
\begin{align} 
    \Omega \coloneq \Omega_{\lambda,\sigma}=\frac{\Id_{\QC_\lambda^d}}{q_\lambda}
 \otimes\sigma_\lambda^{(m)}
 \end{align}
 from \eqref{eq:werner-state-marginal} in terms of permutation operators according to Prop~\ref{prop:fourier-decomp-Werner-states}. Fix $\lambda \vdash_d m$ and write $\sigma\coloneqq \sigma_\lambda^{(m)}$. We may decompose each Schur block in our $m$-qudit Werner state as
\begin{align}
    \Omega = \frac{f_{\lambda}}{ q_{\lambda} m!} \sum_{\pi \in \SC_m} \tr{\sigma R_{\lambda}(\pi^{-1})} P_d (\pi).
\end{align}
As in the symmetric Werner case, we split the sum into two cases: $\pi(m)=m$ and $\pi(m)\neq m$. The latter are entirely supported on $\HC_M$ and thus vanish once we project onto the non-polynomial summands. Fix $\beta \prec \lambda$ and let $S_{\beta}$ denote the projector onto $\QC_{\beta^-}^d \otimes \PC_{\beta}$, so we may write 
\begin{align}
    S = \sum_{\beta \vdash_d m-1} S_{\beta}, \quad \text{where} \quad S_{\beta}S_{\beta'} = \delta_{\beta \beta'} S_{\beta}.
\end{align}
By virtue of the fact that the non-polynomial sector still decomposes fully as a direct sum, it suffices to bound the contribution from each $\beta$. To this end, write
\begin{align}
    S_{\beta} \Omega^{\Gamma_m}S_{\beta} = \frac{f_{\lambda}}{ q_{\lambda} m!} \sum_{\tau \in \SC_{m-1}} \tr{\sigma R_{\lambda}(\tau^{-1})} S_{\beta} P_d (\tau)S_{\beta}
\end{align}
Because the remaining permutations fix the $m$-th site, we may use the Specht branching rules to re-express our coefficients as 
\begin{align}
    \tr{\sigma R_{\lambda}(\tau^{-1})}= \sum_{\gamma \prec \lambda} \tr{\sigma_{\gamma} R_{\gamma}(\tau^{-1})}, \quad \text{where} \quad \sigma_{\gamma} \coloneq J_{\lambda \rightarrow \gamma} \sigma J_{\lambda \rightarrow \gamma}^{\dagger}.
\end{align}
Substituting this into our decomposition yields
\begin{align}
    S_{\beta} \Omega^{\Gamma_m}S_{\beta} = \frac{f_{\lambda}}{ q_{\lambda} m!} \sum_{\gamma \prec \lambda} \sum_{\tau \in \SC_{m-1}} \tr{\sigma_{\gamma} R_{\gamma}(\tau^{-1})} S_{\beta} P_d (\tau)S_{\beta}.
\end{align}
To simplify further, note that, restricted to the range of $S_{\beta}$, the permutation representation takes the form $S_{\beta}P_d(\pi) S_{\beta} = \Id_{\QC_{\beta^-}^d} \otimes R_{\beta}(\tau)$ in the Schur basis. Substituting this back into our decomposition and factoring out the identity on the Weyl module, we obtain
\begin{align}
    S_{\beta} \Omega^{\Gamma_m}S_{\beta} = \frac{f_{\lambda}}{ q_{\lambda} m!} \Id_{\QC_{\beta^-}}^d \otimes \sum_{\gamma \prec \lambda} \sum_{\tau \in \SC_{m-1}} \tr{\sigma_{\gamma} R_{\gamma}(\tau^{-1})} R_{\beta}(\tau).
\end{align}
Due to Schur orthogonality (see Appendix)
\begin{align}
    \sum_{\tau \in \SC_{m-1}} \tr{\sigma_{\gamma} R_{\gamma}(\tau^{-1})} R_{\beta}(\tau) = 
        \frac{(m-1)!}{f_{\beta}} \sigma_{\beta} \delta_{\gamma \beta},
\end{align}
the inner sum collapses to $\gamma = \beta$, yielding
\begin{align}
    S_{\beta} \Omega^{\Gamma_m}S_{\beta} = \frac{f_{\lambda}}{ q_{\lambda} m!} \frac{(m-1)!}{f_{\beta}}\Id_{\QC_{\beta^-}}^d \otimes \sigma_{\beta} = \frac{f_{\lambda}}{m q_{\lambda} f_{\beta}} \Id_{\QC_{\beta^-}}^d \otimes \sigma_{\beta}.
\end{align}
We now must determine the form of $S_{\beta}D S_{\beta}$. From Prop.~\ref{prop:Werner-marginals}, recalling the convention that $\sigma_{\gamma}=0$ if $\gamma \nprec \lambda$, we have
\begin{align}
    D = \left(\bigoplus_{\gamma \vdash_d m-1} \frac{1}{q_{\gamma}} \Id_{\QC_{\gamma}^d} \otimes \sigma_{\gamma} \right) \otimes \Id_{\mathbb{C}^d} \cong \bigoplus_{\gamma \vdash_d m-1} \frac{1}{q_{\gamma}} \Id_{\QC_{\gamma}^d \otimes \mathbb{C}^d} \otimes \sigma_{\gamma},
\end{align}
where we have regrouped factors so we can apply dual Pieri to the Weyl module. Doing so, we may write
\begin{align}
    D = \bigoplus_{\gamma \vdash_d m-1} \frac{1}{q_{\gamma}} \left(\left[\bigoplus_{\alpha \prec \gamma} \Id_{\QC_{\alpha}^d}\right] \oplus \Id_{\QC_{\gamma^-}^d}\right)\otimes \sigma_{\gamma}.
\end{align}
In the mixed Schur basis, $S_{\beta}$ is a block-diagonal matrix whose entries are zero everywhere except on $\QC_{\beta^-}\otimes \PC_{\beta}$, where it is the identity. Thus, projecting via $S_{\beta}$ yields
\begin{align}
    S_{\beta} D S_{\beta} = \frac{1}{q_{\beta}} \Id_{\QC_{\beta^-}^d} \otimes \sigma_{\beta}.
\end{align}
Finally, we can obtain a closed-form expression for the non-polynomial part of $X_{\lambda,\sigma}$ by subtracting these two contributions 
\begin{align}
    S_{\beta} X_{\lambda, \sigma} S_{\beta} &= S_{\beta}\left( \Omega^{\Gamma_m} - \frac{1}{d} D\right)S_{\beta} = \left(\frac{f_{\lambda}}{m q_{\lambda} f_{\beta}} - \frac{1}{d q_{\beta}}\right) \Id_{\QC_{\beta^-}^d} \otimes \sigma_{\beta}.
\end{align}
This completes the real work required to bound the non-polynomial contributions. All that remains is to simplify the coefficient and bound use a content bound as in the symmetric case. Using the one-box identity $m q_{\lambda} f_{\beta} = (d+c(\lambda \setminus \beta)) q_{\beta} f_{\lambda}$, we may write
\begin{align}
    S_{\beta} X_{\lambda, \sigma} S_{\beta} &= \left(\frac{1}{(d+c(\lambda \setminus \beta)) q_{\beta}} - \frac{1}{d q_{\beta}}\right) \Id_{\QC_{\beta^-}^d} \otimes \sigma_{\beta} = \left(\frac{-c(\lambda \setminus \beta)}{d(d+c(\lambda \setminus \beta)) q_{\beta}} \right) \Id_{\QC_{\beta^-}^d} \otimes \sigma_{\beta}.
\end{align}
Taking the trace norm, we obtain
\begin{align}
    \|S_{\beta} X_{\lambda, \sigma} S_{\beta} \|_1 &= \frac{\abs{c(\lambda \setminus \beta)}}{d(d+c(\lambda \setminus \beta)) q_{\beta}} \|\Id_{\QC_{\beta^-}^d} \otimes \sigma_{\beta} \|_1 = \frac{\abs{c(\lambda \setminus \beta)}}{d(d+c(\lambda \setminus \beta)) q_{\beta}} s_{\beta}t_{\beta}.
\end{align}
Putting everything together, we find
\begin{align}
    \|S X_{\lambda, \sigma} S \|_1 &= \sum_{\beta \prec \lambda} \|S_{\beta} X_{\lambda, \sigma} S_{\beta} \|_1,\\
    &= \sum_{\beta \prec \lambda} \frac{\abs{c(\lambda \setminus \beta)}}{d(d+c(\lambda \setminus \beta)) q_{\beta}} s_{\beta}t_{\beta},\\
    &\leq \sum_{\beta \prec \lambda} \frac{\abs{c(\lambda \setminus \beta)}}{(d+c(\lambda \setminus \beta))} t_{\beta}, \quad & s_{\beta} \leq d q_{\beta}\\
    &\leq \sum_{\beta \prec \lambda} \frac{m-1}{(d-m+1) } t_{\beta}, \quad & \abs{c(\lambda \setminus \beta)} \leq m-1\\
    &= \frac{m-1}{(d-m+1) },
\end{align}
where we have used $\sum_{\beta \prec \lambda} t_\beta = 1$ in the last line. Putting all of this together, the original trace norm we sought to bound is given as
\begin{align}
    \|X_{\lambda,\sigma}\|_1 \leq \left(1+\frac{1}{d}\right) \frac{m-1}{d-m+2} + \frac{m-1}{d-m+1}.
\end{align}
We can then insert this bound into our upper bound on the bias
\begin{align}
     \abs{\tr{M(\rho_0-\rho_1)}} 
    &\leq \sum_{m=2}^n \sum_{\lambda \vdash_d m} p_{\lambda}\cdot \left[ \left(1+\frac{1}{d}\right) \frac{m-1}{d-m+2} + \frac{m-1}{d-m+1} \right],\\
    &=  \sum_{m=2}^n  \left[ \left(1+\frac{1}{d}\right) \frac{m-1}{d-(m-2)} + \frac{m-1}{d-(m-1)} \right],\\
    &\leq  \sum_{m=2}^n  \left[ 3 \cdot  \frac{(m-1)}{d} \right], \quad &d \geq m \geq 2 \\
    &= \frac{3}{2}\cdot \frac{n(n-1)}{d},
\end{align}
and we are done.

\end{proof}

\subsection{Worst-case Optimality of Theorem~\ref{thm:multipartite-werner-indistinguishability}}
We now show that there exists a pair of Werner states, along with a simple non-adaptive local (and thus PPT-BOTH) measurement  which achieves a bias $\Omega(n^2/d)$, implying Theorem~\ref{thm:multipartite-werner-indistinguishability} is worst-case optimal in $n,d$. We emphasize that our matching lower bound is achieved by a nonadaptive local measurement followed by classical post-processing, which falls strictly within all allowable LOCC operations. Thus, the PPT-BOTH relaxation preserves the optimal worst-case dependence on the number of parties and local dimension: even restricting to LOCC cannot improve the $O(n^2/d)$ uniform bound beyond constant factors. This is surprising given the fact that the PPT relaxation can lead to drastically looser results, compared to LOCC, in various quantum information processing tasks~\cite{beigi2010Approximating,bennett1999Quantum,cheng2023Discrimination,liu2023Complexity}. 

\begin{proposition}[Non-adaptive, Local Distinguisher] \label{prop:haar-vs-max-mixed-distinguisher}
    There exists a non-adaptive local measurement $\{M,\Id-M\}$ that distinguishes the maximally mixed state on the symmetric subspace of $(\mathbb{C}^d)^{\otimes n}$ from the maximally mixed state on the full space with bias
    \begin{align}
        \tr{M(\sigma^{(n)} - \tau_d^{\otimes n}) } = \frac{(d)_n}{d^n} - \frac{(d)_n}{d(d+1)\dotsm (d+n-1)},
    \end{align}
    where $(d)_n \coloneq d(d-1)\dotsm (d-n+1)$ denotes the falling factorial. When $d \geq 2n(n-1)$, 
    \begin{align}
        \tr{M(\sigma^{(n)} - \tau_d^{\otimes n}) } \geq \Omega\left(\frac{n^2}{d}\right).
    \end{align}
\end{proposition}
\begin{proof}
    Fix an orthonormal basis $\{\ket{i}\}_{i=1}^d$ of $\mathbb{C}^d$. We will output ``pure" if at least 2 of the $n$ outcomes coincide. To analyze this so-called \textit{collision tester}, let us first define a set containing pairwise distinct entries
    \begin{align}
        \IC_{n,d} \coloneq \{(i_1,\dots,i_n) \in [d]^n : i_a \neq i_b \text{ whenever } a\neq b\}.
    \end{align}
    Note that the cardinality of this set is the falling factorial $\abs{\IC_{n,d}} = (d)_n = d(d-1) \dotsm (d-n+1)$. This notation allows us to define a projector corresponding to no collisions as
    \begin{align}
        M_{\rm nc} \coloneq \sum_{\boldsymbol{i} \in \IC_{n,d}} \ketbra{\boldsymbol{i}}{\boldsymbol{i}}.
    \end{align}
    There are either no collisions or at least one collision, thus we may define the projector corresponding to collisions as $ M_{\rm col} \coloneq \Id - M_{\rm nc}.$ Both of these operators are diagonal, thus invariant under partial transpositions. Thus, the measurement $\{M_{\rm col}, M_{\rm nc}\}$ is a PPT-BOTH measurement. It is really a member of a strict subset of all PPT-BOTH measurements, which we might call non-adaptive single-copy measurements.

    It is easiest to think through the probability of no collisions and then translate the results at the end. The probability of no collisions given the maximally mixed state is computed as
    \begin{align}
        \tr{M_{\rm nc} \tau_d^{\otimes n}} &= \frac{1}{d^n}\sum_{\boldsymbol{i} \in \IC_{n,d}} \tr{\ketbra{\boldsymbol{i}}{\boldsymbol{i}} \Id_d^{\otimes n}} = \frac{\abs{\IC_{d,n}}}{d^n} = \frac{(d)_n}{d^n}.
    \end{align}
    If instead we are actually given the Haar random pure state, the probability of no collisions can be computed as follows. Recall that 
    \begin{align}
        \sigma^{(n)} = \frac{\Pi_n}{D_n}= \frac{1}{D_n n!} \sum_{\pi \in \SC_n} P_d(\pi).
    \end{align}
    But, by definition, $\IC_{n,d}$ contains pairwise distinct strings, thus only $\pi = e$ yields a non-zero contribution. We obtain
    \begin{align}
        \tr{M_{\rm nc} \sigma^{(n)}} &= \frac{1}{D_n n!} \sum_{\boldsymbol{i} \in \IC_{n,d}} \sum_{\pi \in \SC_n}\tr{P_d (\pi) \ketbra{\boldsymbol{i}}{\boldsymbol{i}}} = \frac{\abs{\IC_{n,d}}}{D_n n!}= \frac{(d)_n}{d(d+1)\dotsm (d+n-1)}.
    \end{align}
    Thus, the bias for our collision tester is given as 
    \begin{align}
        \tr{M_{\rm col} (\sigma^{(n)} - \tau_d^{\otimes n})} =  \tr{M_{\rm nc} (\tau_d^{\otimes n} - \sigma^{(n)})} = \frac{(d)_n}{d^n} - \frac{(d)_n}{d(d+1)\dotsm (d+n-1)}.
    \end{align}
    Let us call the bias $b$. We can factor out the common term to write \begin{align}
        b = \frac{(d)_n}{d^n}\left(1- \frac{d^n}{d(d+1)\cdots(d+n-1)}\right) =  \left[\prod_{k=0}^{n-1}\left(1-\frac{k}{d}\right)\right]\left[1-\prod_{k=0}^{n-1}\left(1-\frac{k}{d+k}\right)\right]
    \end{align} 
    For simplicity, Weierstrass's product inequality---$\prod_{k=0}^{n-1} (1-x_k) \geq 1 - \sum_{k=0}^{n-1}x_k$ for any $x_0,\dots,x_{n-1} \in [0,1]$---allows us to lower bound the first term with $3/4$. For the second term, the classical inequality $1-x \leq e^{-x}$ bounds the product in the second term \begin{align}
        1-\prod_{k=0}^{n-1}\left(1-\frac{k}{d+k}\right) \geq 1- \exp\left(-\sum_{k=0}^{n-1} \frac{k}{d+k}\right) \geq 1-\exp\left(-\frac{n(n-1)}{3d}\right),
    \end{align}
    where in the last step we use $\frac{k}{d+k} \geq \frac{k}{\frac{3}{2}d} = \frac{2k}{3d}$ since $k \leq n-1 \leq n(n-1) \leq \frac{d}{2}$. Finally, we can use $1 -e^{-y} \geq y/2$ which is valid for $y = \frac{n(n-1)}{3d} \leq \frac{1}{6}$. This completes the bound on the bias \begin{align}
        b \geq \frac{3}{4}\left(1-\exp\left(-\frac{n(n-1)}{3d}\right)\right) \geq  \Omega(n^2/d).
    \end{align}
\end{proof}
With the tightness of Theorem~\ref{thm:multipartite-werner-indistinguishability} established, we turn to the implications this result has for multipartite data hiding.

\subsection{Data Hiding with Multipartite Werner States}
Because our bound holds uniformly for all Werner states, it allows many data-hiding pairs to be constructed simultaneously. Recall that every Werner state has Schur-basis form
\begin{align}
 \rho = \bigoplus_{\lambda\vdash_d n} p_\lambda\frac{\Id_{\QC_\lambda^d}}{q_\lambda} \otimes\sigma_\lambda,
\end{align}
where each $\sigma_\lambda$ is an arbitrary density operator on $\PC_\lambda$. To construct mutually orthogonal hiding states, we choose an orthonormal basis $\{\ket{a}\}_{a=1}^{f_\lambda}$ of each Specht module and take $\Omega_{\lambda,\sigma}$ with $\sigma=\ketbra{a}{a}$. These states have mutually orthogonal supports, both within each $\lambda$-block and between different blocks. A single global measurement can therefore distinguish all
\begin{align}
 N=\sum_{\lambda\vdash_d n}f_\lambda
\end{align}
states perfectly, giving $\binom{N}{2}$ pairs within the same family, which is maximal because orthogonal Werner states must occupy mutually orthogonal supports in the Specht spaces, whose total dimension is $N$. For $d\ge n$, every partition of $n$ occurs, and $\sum_{\lambda\vdash_d n}f_\lambda^2=n!$. Consequently,
\begin{align}
 \sqrt{n!}\le N\le n!
 \implies N=2^{\Theta(n\log n)}.
\end{align}
This straightforward counting argument is unchanged from the multipartite data hiding section in Harrow's work on approximate orthogonality~\cite{harrow2023Approximate}; our stronger bound improves the security guarantee for every pair in this family. To make this precise, fix $0<\epsilon<1$. Theorem~\ref{thm:multipartite-werner-indistinguishability} guarantees PPT-BOTH distinguishing bias at most $\epsilon$ whenever
\begin{align}
 n(n-1)\le \frac{2\epsilon d}{3}.
\end{align}
Thus, we may take $n=\Theta(\sqrt{\epsilon d})$. By comparison, Harrow's bound $6n^2/\sqrt d$ guarantees the same security when
\begin{align}
 n^2\le \frac{\epsilon\sqrt d}{6},
\end{align}
allowing $n=\Theta(\sqrt{\epsilon}\,d^{1/4})$. At fixed $\epsilon$, our result therefore increases the size of the certified family from $2^{\Theta(d^{1/4}\log d)}$ to $2^{\Theta(\sqrt d\log d)}$ perfectly distinguishable messages, with every pair having PPT-BOTH bias at most $\epsilon$. Proposition~\ref{prop:haar-vs-max-mixed-distinguisher} implies worst-case optimality of Theorem~\ref{thm:multipartite-werner-indistinguishability} and thus precludes any improvement on this uniform security guarantee; however, particular pairs of Werner states may remain secure well beyond this regime.

Indeed, particular Werner-state pairs can hide a classical bit substantially better. For even $n=2r$ and $d\ge2$, Lancien and Winter~\cite[Thm.~9]{lancien2013Distinguishing} consider the permutation
\begin{align}
 \pi=(1,r+1)(2,r+2)\cdots(r,2r).
\end{align}
The operator $P_d(\pi)$ swaps the first and last $r$ registers. Its normalized positive and negative eigenspace projectors,
\begin{align}
 \rho_0=\frac{\Id+P_d(\pi)}{d^{2r}+d^r},
 \qquad
 \rho_1=\frac{\Id-P_d(\pi)}{d^{2r}-d^r},
\end{align}
are orthogonal Werner states because $P_d(\pi)$ commutes with $U^{\otimes n}$ for every $U\in\UC_d$. To see why they are difficult to distinguish locally, group the registers into two subsystems of dimension $D=d^r$.Writing $A=\{1,\ldots,r\}$ and $B=\{r+1,\ldots,2r\}$, partial transpose on $A$ sends $P_d(\pi)$ to $\psi^+_{AB}$, the unnormalized maximally entangled operator on the two $D$-dimensional subsystems. Thus, taking the partial transpose of the first $r$ registers yields
\begin{align}
 (\rho_0-\rho_1)^{\Gamma_A}
 =
 \frac{2}{D^2-1}
 \left(\psi^+_{AB}-\frac{\Id}{D}\right).
\end{align}
Consequently, any PPT-BOTH measurement $\{M,\Id-M\}$ satisfies
\begin{align}
 \abs{\tr{M(\rho_0-\rho_1)}}
 &\le \frac12
 \left\|(\rho_0-\rho_1)^{\Gamma_S}\right\|_1
 =\frac{2}{d^r}
 =O(d^{-n/2}).
\end{align}
This construction illustrates the substantial variation in hiding strength within the Werner-state family: our bound determines the optimal uniform security scaling, while additional structure can make particular pairs far harder to distinguish.

\section{Discussion and Future Directions}
In this work, we have established the optimal dependence of the uniform security guarantee for multipartite Werner states on the number of parties and local dimension, resolving a longstanding question about the Eggeling–Werner data-hiding scheme. Our $O(n^2/d)$ upper bound on the PPT-BOTH distinguishing bias is matched by an explicit pair and a simple nonadaptive local measurement followed by classical post-processing. Thus, despite allowing a substantially broader class of measurements than LOCC, the PPT-BOTH relaxation preserves the optimal worst-case scaling. A direct analysis of LOCC could improve the constants, but cannot improve the dependence on $n$ and $d$ in a bound valid for every pair of Werner states. At fixed security, our result extends the certified hiding regime from $n=O(d^{1/4})$ to $n=O(\sqrt d)$.

The proof also establishes a connection between restricted-measurement discrimination and representation-theoretic methods developed for port-based teleportation. A telescoping argument reduces the distinguishing bias to trace norms of partially transposed operators, whose structure is then determined by mixed Schur–Weyl duality. For symmetric Werner states, this structure permits an exact spectral calculation; for arbitrary Werner states, an operator-order inequality controls the coherences that remain between equivalent irreducible representations. Interpreting the tensor factors as copy registers translates the same indistinguishability bound into sample complexity lower bounds for testing unitarily invariant properties, recovering the purity-testing lower bound and strengthening the known low-rank-testing bound.

We hope that the combination of telescoping reductions, partial transposition, and mixed Schur–Weyl duality provides a useful route to analyzing other restricted-measurement protocols. The connection to port-based teleportation suggests that tools developed for one quantum information task may have considerable value in data hiding, property testing, and the many other settings where measurement restrictions determine what can be learned about a quantum system.

\medskip
\noindent \textit{Acknowledgments.}
OpenAI's GPT-6 Astra was used in the development of the paper's main results. In addition to helping check results, it suggested the use of the telescoping sum as well as the use of Claim~\ref{claim:partial-transpose-CP}, two crucial ingredients in our proofs. The authors take full responsibility for the correctness, exposition, and attribution in the final manuscript. The authors thank Louis Schatzki, Luke Coffman, and Graeme Smith for input on the results. JB was supported by a National Science Foundation Mathematical Sciences Postdoctoral Research Fellowship under Award No.~2402287 as well as an IQUIST Postdoctoral Fellowship. FL was supported by National Science Foundation Grant No.~2442410.

\bibliography{main.bib}

\appendix

\section{Additional Proof Details}
\subsection{Proofs of Technical Lemmas}
\label{app:technical-lemma-proofs}
\begin{proposition}[Fourier Expansion of Werner States, Prop.~\ref{prop:fourier-decomp-Werner-states} Resetated] 
    Let $\rho$ be a multipartite Werner state on $(\mathbb{C}^d)^{\otimes n}$. Then, $\rho$ lies in the linear span of permutation operators and admits the expansion
    \begin{align}
        \rho &= \frac{1}{n!} \sum_{\pi \in \SC_n} \left(\sum_{\lambda \vdash_d n} \frac{p_{\lambda} f_{\lambda}}{q_{\lambda}} \tr{\sigma_{\lambda} R_{\lambda}(\pi^{-1})}\right)P_d(\pi).
    \end{align}
    If $\rho$ is a symmetric multipartite Werner state, this expression simplifies to 
    \begin{align}
        \rho &= \frac{1}{n!} \sum_{\pi \in \SC_n} \left(\sum_{\lambda \vdash_d n}\frac{p_{\lambda}}{q_{\lambda}} \chi_{\lambda}(\pi^{-1})\right) P_d(\pi),
    \end{align}
    where $\chi_{\lambda}(\pi) = \tr{R_{\lambda}(\pi)}$ is the character value of $\pi$.
\end{proposition}
\begin{proof}
        The action of $P_d(\pi)$ can be split using Schur-Weyl duality as \begin{align}
        P_d(\pi) = \bigoplus_{\lambda \vdash_d n} \Id_{\QC_\lambda^d} \otimes R_\lambda(\pi).
    \end{align}
    The irreducible representations $R_\lambda(\pi)$ satisfy the matrix-coefficient orthogonality relation \begin{align}
        \frac{f_\lambda}{n!}\sum_{\pi \in \SC_n}R_\mu(\pi)_{ij}R_\lambda(\pi^{-1})_{lk} = \delta_{\mu\lambda}\delta_{ik}\delta_{jl}.
    \end{align}
    Recall from Definition~\ref{def:Werner-states} that \begin{align}
        \rho = \bigoplus_{\mu \vdash_d n}\Id_{\QC_\mu^d}\otimes \frac{p_\mu}{q_\mu}\sigma_\mu.
    \end{align}
    We claim $\rho$ is given by \begin{align}
        \rho = \frac{1}{n!}\sum_{\pi \in \SC_n}\left(\sum_{\lambda\vdash_d n}\frac{p_\lambda f_\lambda}{q_\lambda}\tr{\sigma_\lambda R_\lambda(\pi^{-1})}\right)P_d(\pi)
    \end{align}
    To verify this, we can decompose the right-hand side with Schur-Weyl duality: \begin{align}
        \frac{1}{n!}\sum_{\pi \in \SC_n}\left(\sum_{\lambda\vdash_d n}\frac{p_\lambda f_\lambda}{q_\lambda}\tr{\sigma_\lambda R_\lambda(\pi^{-1})}\right)P_d(\pi) &= \bigoplus_{\mu \vdash_d n}\Id_{\QC_\mu^d}\otimes \frac{1}{n!}\sum_{\pi \in \SC_n}\left(\sum_{\lambda\vdash_d n}\frac{p_\lambda f_\lambda}{q_\lambda}\tr{\sigma_\lambda R_\lambda(\pi^{-1})}\right)R_\mu(\pi)
    \end{align}
    It therefore suffices to show each block with respect to $\mu$ is equivalent to $\frac{p_\mu}{q_\mu}\sigma_\mu$. For the $(i,j)$-th coefficient, we have  \begin{align*}
         \frac{1}{n!}\sum_{\pi \in \SC_n}\left(\sum_{\lambda\vdash_d n}\frac{p_\lambda f_\lambda}{q_\lambda}\tr{\sigma_\lambda R_\lambda(\pi^{-1})}\right)(R_\mu(\pi))_{ij} & = \sum_{\lambda \vdash_d n} \frac{p_\lambda}{q_\lambda}\frac{f_\lambda}{n!}\sum_{\pi\in \SC_n}\left(\sum_{k,l}(\sigma_\lambda)_{kl}R_\lambda(\pi^{-1})_{lk}R_\mu(\pi)_{ij}\right), \\
         & =\sum_{\lambda \vdash_d n} \frac{p_\lambda}{q_\lambda}\sum_{k,l}(\sigma_\lambda)_{kl}\left(\frac{f_\lambda}{n!}\sum_{\pi\in \SC_n}R_\mu(\pi)_{ij}R_\lambda(\pi^{-1})_{lk} \right), \\
         & = \sum_{\lambda \vdash_d n} \frac{p_\lambda}{q_\lambda}\sum_{k,l}(\sigma_\lambda)_{kl}\left(\delta_{\mu\lambda}\delta_{ik}\delta_{jl} \right), \\
         & = \frac{p_\mu}{q_\mu}(\sigma_\mu)_{ij}.
    \end{align*}
    For the second part of the proposition, assume that $\rho$ is a symmetric multipartite Werner state. Schur-Weyl duality tells us that $\sigma_\lambda = \frac{\Id_{\PC_\lambda}}{f_\lambda}$. After substituting, we have \begin{align}
        \rho & = \frac{1}{n!}\sum_{\pi \in \SC_n}\left(\sum_{\lambda\vdash_d n}\frac{p_\lambda f_\lambda}{q_\lambda}\tr{\frac{\Id_{\PC_\lambda}}{f_\lambda} R_\lambda(\pi^{-1})}\right)P_d(\pi), \\
        & = \frac{1}{n!}\sum_{\pi \in \SC_n}\left(\sum_{\lambda\vdash_d n}\frac{p_\lambda}{q_\lambda}\tr{R_\lambda(\pi^{-1})}\right)P_d(\pi),
    \end{align} 
    and denoting $\chi_\lambda(\pi) = \tr{R_\lambda(\pi)}$ completes the proof.
\end{proof}

\subsection{Detailed Proof of Theorem~\ref{thm:symmetric-Werner-distinguishing}}\label{app:proofs-of-main-results}
\begin{proof}
	Throughout the proof, we write $\Pi_\lambda^{(m)}$ for the isotypical projector associated to the Young diagram $\lambda\vdash_d m$ in the tensor representation $\pi\mapsto P_d(\pi)$ of $\SC_m$ on $(\mathbb{C}^d)^{\otimes m}$, adding the super script $(m)$ to keep track of how many systems the projectors acts on.
	
    Our goal is to compute the coefficients $x_{\beta,i}$ of the operator $X_\lambda=\frac{1}{q_\lambda f_\lambda} \left(\Pi_\lambda^{\Gamma_m} - \trc_m \Pi_\lambda \otimes \tau_d\right)$ defined via the diagonalization in \eqref{eq:X_lambda-decomposition}.
    Both $\Pi_\lambda^{\Gamma_m}$ and $\trc_m \Pi_\lambda \otimes \tau_d$ are diagonal themselves in this basis:
    \begin{align}
    	\Pi_\lambda^{\Gamma_m} &= \bigoplus_{\beta\vdash_d m-1} \,\bigoplus_{i\colon \beta_i>\beta_{i+1}} y_{\beta,i} \Id_{\cQ_{\alpha^{(i)}}^d}\otimes \Id_{\cP_\beta} \label{eq:y-coefficients}\\
    	\trc_m \Pi_\lambda \otimes \tau_d &= \bigoplus_{\beta\vdash_d m-1} \,\bigoplus_{i\colon \beta_i>\beta_{i+1}} z_{\beta,i} \Id_{\cQ_{\alpha^{(i)}}^d}\otimes \Id_{\cP_\beta}. \label{eq:z-coefficients}
    \end{align}
    Once the $y_{\beta,i}$ and $z_{\beta,i}$ are determined, the coefficients $x_{\beta,i}$ can be computed as $x_{\beta,i} = (y_{\beta,i}-z_{\beta,i})/(q_\lambda f_\lambda)$.
    
    We start with computing the coefficients $y_{\beta,i}$.
    To this end, we recall the character formula for the isotypical projector $\Pi_\lambda^{(m)}$,
    \begin{align}
        \Pi_\lambda^{(m)} = \frac{f_\lambda}{m!} \sum_{\pi\in S_m} \chi_\lambda(\pi) P_d(\pi),
    \end{align}
    where $\chi_\lambda(\pi)$ is the irreducible character associated with the $S_m$-irrep $\cP_\lambda$.
    Since we will take the partial transpose with respect to the last system, it is useful to split up the sum over $\pi\in S_m$ into two parts: a summation over the subgroup $G=\lbrace \pi\in S_m:\pi(m)=m\rbrace\cong \SC_{m-1}$ of permutations fixing $m$, and a summation over the $m-1$ cosets of $G$ in $S_m$ with transversal $\lbrace (1,m),\dots,(m-1,m)\rbrace$:
    \begin{align}
        \Pi_\lambda^{(m)} = \frac{f_\lambda}{m!} \sum_{\pi \in G} \chi_\lambda(\pi) P_d'(\pi)\otimes \Id + \frac{f_\lambda}{m!} \sum_{k=1}^{m-1} \sum_{\pi\in G} \chi_\lambda((k,m)\pi) P_d(k,m) P_d'(\pi)\otimes \Id_m
    \end{align}
    where for $\pi\in G$ we denote by $P_d'(\pi)$ the effective permutation operator acting only on the first $m-1$ systems such that $P_d(\pi) = P_d'(\pi)\otimes \Id_m$ for such permutations.
    Taking the partial transpose over the last system,
    \begin{align}
        \left(\Pi_\lambda^{(m)}\right)^{\Gamma_m} &= \frac{f_\lambda}{m!} \sum_{\pi \in G} \chi_\lambda(\pi) P_d'(\pi)\otimes \Id_m^T + \frac{f_\lambda}{m!} \sum_{k=1}^{m-1} \sum_{\pi\in G} \chi_\lambda((k,m)\pi) \left(P_d(k,m)\right)^{\Gamma_m} \left(P_d'(\pi)\otimes \Id_m\right)\\
        &= \frac{f_\lambda}{m!} \sum_{\pi \in G} \chi_\lambda(\pi) P_d'(\pi)\otimes \Id_m + \frac{f_\lambda}{m!} \sum_{k=1}^{m-1} \sum_{\pi\in G} \chi_\lambda((k,m)\pi) \psi^+_{km} \left( P_d'(\pi)\otimes \Id_m \right),
        \label{eq:expanded-isotypical-projector}
    \end{align}
    where we used that $\left(P_d(k,m)\right)^{\Gamma_m}=\psi^+_{km}=\sum_{i,j} |i\rangle\langle j|_k\otimes |i\rangle\langle j|_m$ is the unnormalized maximally entangled state on systems $k$ and $m$, and we are omitting identity operators acting on the remaining $m-2$ systems.

    Using the branching rule for symmetric groups,
    \begin{align}
        \chi_\lambda\big|^{\SC_m}_{\SC_{m-1}} = \sum_{\mu\prec \lambda} \chi_\mu,
        \label{eq:branching-rule2}
    \end{align}
    the first sum in \eqref{eq:expanded-isotypical-projector} can be rewritten as
    \begin{align}
        \frac{f_\lambda}{m!} \sum_{\pi \in G} \chi_\lambda(\pi) P_d'(\pi)\otimes \Id_m &= \frac{f_\lambda}{m!} \sum_{\pi \in G} \sum_{\mu\prec\lambda} \chi_\mu(\pi) P_d'(\pi)\otimes \Id_m\\
        &= \frac{f_\lambda}{m} \sum_{\mu\prec \lambda} \frac{1}{f_\mu} \frac{f_\mu}{(m-1)!} \sum_{\pi\in G} \chi_\mu(\pi) P_d'(\pi)\otimes \Id_m\\
        &= \frac{f_\lambda}{m} \sum_{\mu\prec \lambda} \frac{1}{f_\mu} \Pi^{(m-1)}_\mu \otimes \Id_m. \label{eq:sum-over-G-rewritten}
    \end{align}
    
    \newcommand{\ch}[1]{\mathbf{1}_{#1}}
    We now proceed with computing the coefficients $y_{\beta,i}$ in \eqref{eq:y-coefficients}.
    First, we recall that we write 
    \begin{align} 
    	\alpha^{(i)} = (\beta_1,\dots,\beta_{i-1},\beta_i-1,\beta_{i+1},\dots,\beta_d)
    \end{align}  
    for the weight obtained from the partition $\beta=(\beta_1,\dots,\beta_d)\vdash_d m-1$ by removing 1 in position $i$, provided $\beta_i>\beta_{i+1}$.
    We denote by $\Theta_{\alpha^{(i)}}^d$ the isotypical projectors corresponding to the weight $\alpha^{(i)}$ for the representation $U\mapsto U^{\otimes (m-1)}\otimes \overline{U}$ of $U(d)$, and we write $\ch{\mathsf{P}}$ for the indicator function of the logical statement $\mathsf{P}$.
    Then on the one hand, we get from \eqref{eq:y-coefficients} that
    \begin{align}
    	\trc\left(\left(\Pi^{(m)}_\lambda\right)^{\Gamma_m} \left(\Pi^{(m-1)}_\beta\otimes \Id_m\right) \Theta_{\alpha^{(i)}}^d\right) = y_{\beta,i} f_\beta \,q_{\alpha^{(i)}}.
    	\label{eq:coefficients-times-dimensions}
    \end{align}
    On the other hand, using \eqref{eq:expanded-isotypical-projector} and \eqref{eq:sum-over-G-rewritten},
    \begin{multline}
    	\trc\left(\left(\Pi^{(m)}_\lambda\right)^{\Gamma_m} \left(\Pi^{(m-1)}_\beta\otimes \Id_m\right) \Theta_{\alpha^{(i)}}^d\right)
    	= \frac{f_\lambda}{m} \sum_{\mu\prec \lambda} \frac{1}{f_\mu}  \trc\left(\left(\Pi^{(m-1)}_\mu \otimes \Id_m\right) \left( \Pi^{(m-1)}_\beta\otimes \Id_m\right)  \Theta_{\alpha^{(i)}}^d\right)\\
    	{} + \frac{f_\lambda}{m!} \sum_{k=1}^{m-1} \sum_{\pi\in G} \chi_\lambda((k,m)\pi) \trc\left( \psi^+_{km} \left( P_d'(\pi)\otimes \Id_m \right) \left(\Pi^{(m-1)}_\beta\otimes \Id_m\right) \Theta_{\alpha^{(i)}}^d\right). \label{eq:y-coefficients-expanded}
    \end{multline}
    For the first sum in \eqref{eq:y-coefficients-expanded}, orthogonality of the projectors $\Pi^{(m-1)}_*$ gives
    \begin{align} 
    	\frac{f_\lambda}{m} \sum_{\mu\prec \lambda} \frac{1}{f_\mu}  \trc\left(\left(\Pi^{(m-1)}_\mu \otimes \Id_m\right) \left( \Pi^{(m-1)}_\beta\otimes \Id_m\right)  \Theta_{\alpha^{(i)}}^d\right) &= \ch{\beta\prec\lambda}\frac{f_\lambda}{mf_\beta} \trc\left( \left( \Pi_{\beta}^{(m-1)}\otimes \Id_m\right) \Theta_{\alpha^{(i)}}^d\right)\\
    	&= \ch{\beta\prec\lambda}\frac{f_\lambda q_{\alpha^{(i)}} }{m}.
    	\label{eq:S_(m-1)-sum}
    \end{align}
    
    For the second sum in \eqref{eq:y-coefficients-expanded}, we use the same argument as in \cite[App.~A]{christandl2021Asymptotic}: The operator $\psi^+_{km}$ is invariant under $U\otimes \overline{U}$, and hence on the support of $\psi^+$ the action of $U(d)$ via $U^{\otimes (m-1)}\otimes \overline{U}$ coincides with the action via $U^{\otimes (m-2)}$, and hence
    \begin{align}
    	\psi^+_{km} \Theta_{\alpha^{(i)}}^d = \ch{\alpha^{(i)}\vdash_d m-2} \psi^+_{km} \left(\Pi^{(m-2)}_{\alpha^{(i)}} \otimes \Id_{km}\right),\label{eq:effective-action-on-psi}
    \end{align}
    where $\ch{\alpha^{(i)}\vdash_d m-2}$ reflects the fact that this is non-zero only if $\alpha^{(i)}$ is a valid partition satisfying $(\alpha^{(i)})_{d}\geq 0$.
    Using \eqref{eq:effective-action-on-psi}, the second sum in \eqref{eq:y-coefficients-expanded} becomes
    \begin{align}
    	& \frac{f_\lambda}{m!} \sum_{k=1}^{m-1} \sum_{\pi\in G} \chi_\lambda((k,m)\pi) \trc\left(\left( \psi^+_{km} \left( P_d'(\pi)\otimes \Id_m \right) \right) \left(\Pi^{(m-1)}_\beta\otimes \Id_m\right) \Theta_{\alpha^{(i)}}^d\right) \notag\\
    	&\qquad {} = \ch{\alpha^{(i)}\vdash_d m-2} \frac{f_\lambda}{m!} \sum_{k=1}^{m-1} \sum_{\pi\in G} \chi_\lambda((k,m)\pi) \trc\left( \psi^+_{km} \left( P_d'(\pi)\otimes \Id_m \right) \left(\Pi^{(m-1)}_\beta\otimes \Id_m\right) \left(\Pi^{(m-2)}_{\alpha^{(i)}} \otimes \Id_{km}\right)\right) \label{eq:take-partial-trace}\\
    	&\qquad {} = \ch{\alpha^{(i)}\vdash_d m-2} \frac{f_\lambda}{m!} \sum_{k=1}^{m-1} \sum_{\pi\in G} \chi_\lambda((k,m)\pi) \trc\left(  P_d'(\pi) \Pi^{(m-1)}_\beta \left(\Pi^{(m-2)}_{\alpha^{(i)}} \otimes \Id_{k}\right)\right)\\
        &\qquad {} = \ch{\alpha^{(i)}\vdash_d m-2} \frac{f_\lambda}{m!} \sum_{k=1}^{m-1} \sum_{\pi\in G} \chi_\lambda((k,m)\pi) \trc\left(  P_d'(\pi) \Pi^{(m-1)}_\beta P_d'(k,m-1) \left(\Pi^{(m-2)}_{\alpha^{(i)}} \otimes \Id_{m-1}\right)P_d'(k,m-1)\right)\\
        &\qquad {} = \ch{\alpha^{(i)}\vdash_d m-2} \frac{f_\lambda}{m!} \sum_{k=1}^{m-1} \sum_{\pi\in G} \chi_\lambda((k,m)\pi) \trc\left(  P_d'((k,m-1)\pi(k,m-1)) \Pi^{(m-1)}_\beta \left(\Pi^{(m-2)}_{\alpha^{(i)}} \otimes \Id_{m-1}\right)\right) \\
        &\qquad {} = \ch{\alpha^{(i)}\vdash_d m-2} \frac{f_\lambda}{m!} \sum_{k=1}^{m-1} \sum_{\pi\in G} \chi_\lambda((k,m)(k,m-1)\pi(k,m-1)) \trc\left(  P_d'(\pi) \Pi^{(m-1)}_\beta  \left(\Pi^{(m-2)}_{\alpha^{(i)}} \otimes \Id_{m-1}\right)\right) \label{eq:conj-iso}\\
        &\qquad {} = \ch{\alpha^{(i)}\vdash_d m-2} \frac{f_\lambda (m-1)}{m!} \sum_{\pi\in G} \chi_\lambda((m-1,m)\pi) \trc\left(  P_d'(\pi) \Pi^{(m-1)}_\beta  \left(\Pi^{(m-2)}_{\alpha^{(i)}} \otimes \Id_{m-1}\right)\right) \label{eq:transposition-product}\\
        &\qquad {} = \ch{\alpha^{(i)}\vdash_d m-2} \frac{f_\lambda (m-1)}{m} \sum_{\mu\prec\lambda}\frac{1}{f_\mu}\sum_{\nu\prec\mu} \frac{1}{c(\lambda\setminus\mu)-c(\mu\setminus\nu)}  \trc\left(  \left(\Pi^{(m-2)}_\nu \Pi^{(m-2)}_{\alpha^{(i)}}\otimes \Id_{m-1}\right)\Pi^{(m-1)}_\mu \Pi^{(m-1)}_\beta \right) \label{eq:apply-lemma}\\
        &\qquad {} = \ch{\beta\prec\lambda \wedge \alpha^{(i)}\vdash_d m-2} \frac{m-1}{m} \frac{f_\lambda f_{\alpha^{(i)}}}{f_\beta} \frac{q_\beta}{c(\lambda\setminus\beta)-c(\beta\setminus\alpha^{(i)})},
        \label{eq:other-sum}
    \end{align}
    where in \eqref{eq:take-partial-trace} we took the partial trace over the $m$-th system, in \eqref{eq:conj-iso} we used the fact that $\pi\mapsto (k,m-1)\pi (k,m-1)$ is an isomorphism (and hence a bijection) on $G$, in \eqref{eq:transposition-product} we used cyclicity of $\chi_\lambda$ and the identity $(k,m-1)(k,m)(k,m-1)=(m-1,m)$, and in \eqref{eq:apply-lemma} we applied Lemma~\ref{lem:character-formula}.
    
    Using the expressions \eqref{eq:S_(m-1)-sum} and \eqref{eq:other-sum} for the two sums in \eqref{eq:y-coefficients-expanded} and dividing by $f_\beta q_{\alpha^{(i)}}$ according to \eqref{eq:coefficients-times-dimensions} determines the coefficients $y_{\beta,i}$ of $(\Pi_\lambda)^{\Gamma_m}$ in \eqref{eq:y-coefficients} as
    \begin{align}
    	y_{\beta,i} &= \ch{\beta\prec\lambda}\frac{f_\lambda }{m f_\beta} + \ch{\beta\prec\lambda \wedge \alpha^{(i)}\vdash_d m-2} \frac{m-1}{m} \frac{f_\lambda f_{\alpha^{(i)}}}{f_\beta^2 } \frac{q_\beta}{q_{\alpha^{(i)}}} \frac{1}{c(\lambda\setminus\beta)-c(\beta\setminus\alpha^{(i)})}. \label{eq:y-coefficients-formula}
    \end{align}
    
    It remains to determine the coefficients $z_{\beta,i}$ for the operator $\trc_m\Pi_\lambda\otimes \tau_d$ defined in \eqref{eq:z-coefficients}.
    To this end, we use a formula for the partial trace over an isotypical projector $\Pi_\lambda$ derived in \cite{christandl2007oneandahalf}:
    \begin{align}
    	\trc_m\Pi^{(m)}_\lambda = q_\lambda \sum_{\mu\prec\lambda}\frac{1}{q_\mu} \Pi^{(m-1)}_\mu.
    \end{align}
    Using this identity, we compute:
    \begin{align}
    	f_\beta q_{\alpha^{(i)}} z_{\beta,i} &= \trc\left(\left(\trc_m\Pi^{(m)}_\lambda\otimes \tau_d\right)\left(\Pi^{(m-1)}_\beta\otimes\Id_m\right)\Theta_{\alpha^{(i)}}^d\right)\\
    	&= \frac{q_\lambda}{d} \sum_{\mu\prec\lambda} \frac{1}{q_\mu} \trc\left(\left(\Pi^{(m-1)}_\mu\Pi^{(m-1)}_\beta\otimes\Id_m\right) \Theta_{\alpha^{(i)}}^d\right)\\
    	&= \ch{\beta\prec\lambda} \frac{q_\lambda}{dq_\mu}f_\beta q_{\alpha^{(i)}},\\
    	\intertext{and thus}
    	z_{\beta,i} &	= \ch{\beta\prec\lambda} \frac{q_\lambda}{dq_\mu}.\label{eq:z-coefficients-formula}
    \end{align}
    The final result is obtained by using \eqref{eq:y-coefficients-formula} and \eqref{eq:z-coefficients-formula} in $x_{\beta,i} = (y_{\beta,i}-z_{\beta,i})/(q_\lambda f_\lambda)$, concluding the proof.
\end{proof}

In the proof above we make use of the following technical lemma.

\begin{lemma}\label{lem:character-formula}
	Let $\lambda\vdash_d m$ and $G = \lbrace\pi\in S_m : \pi(m)=m\rbrace \cong \SC_{m-1}$.
    Let $P_d'(\pi)$ be the operator such that $P_d(\pi) = P_d'(\pi)\otimes \Id_m$ on $(\mathbb{C}^d)^{\otimes m}$.
	Then,
	\begin{align}
		\sum_{\pi\in G} \chi_\lambda((m-1,m)\pi) P_d'(\pi) = (m-1)! \sum_{\mu\prec\lambda}\frac{1}{f_\mu}\sum_{\nu\prec\mu} \frac{1}{c(\lambda\setminus\mu)-c(\mu\setminus\nu)} \left(\Pi^{(m-2)}_\nu\otimes \Id_{m-1}\right)\Pi^{(m-1)}_\mu.
	\end{align}
\end{lemma}
\begin{proof}
	Throughout the proof we abbreviate $\tau=(m-1,m)$.
	The subgroup $G$ is isomorphic to $\SC_{m-1}$, and hence the branching rule \eqref{eq:branching-rule2} implies that $p_\lambda(\pi) = \oplus_{\mu\prec\lambda} p_\mu(\pi)$.
	The character value $\chi_\lambda(\tau \pi)$ can thus be computed as
	\begin{align}
		\chi_\lambda(\tau \pi) = \sum_{\mu\prec\lambda} \trc_{\cP_\mu}\left(\tau^\lambda_\mu p_\mu(\pi)\right) = \sum_{\mu\prec\lambda} \sum_{i,j} (\tau_\mu^\lambda)_{ij}p_\mu(\pi)_{ji},
	\end{align}
	where $\tau_\mu^\lambda$ is the projection of $p_\lambda(\tau)$ onto the irrep $\cP_\mu$ (which we will compute explicitly below), and for each $\mu\prec\lambda$ the sum in the second equality is over some basis of $\cP_\mu$.
	We also consider the isotpyical decomposition of $(\mathbb{C}^d)^{\otimes m-1}$ with respect to the tensor representation of $\SC_{m-1}$, 
	\begin{align}
		(\mathbb{C}^d)^{\otimes m-1} &\cong \bigoplus_{\omega\vdash_d m-1} \cP_\omega \otimes \cQ_\omega^d\\
        \intertext{and the corresponding action of $\pi\in \SC_{m-1}$ as}
		P_d'(\pi) &\cong \bigoplus_{\omega\vdash_d m-1} p_\omega(\pi) \otimes \Id_{\cQ_\omega^d}.
	\end{align}
	
	Substituting these expressions, we compute:
	\begin{align}
		\sum_{\pi\in G} \chi_\lambda(\tau\pi) P_d'(\pi) &= \sum_{\pi\in \SC_{m-1}} \sum_{\mu\prec\lambda} \sum_{i,j} (\tau_\mu^\lambda)_{ij}p_\mu(\pi)_{ji} \bigoplus_{\omega\vdash_d m-1} p_\omega(\pi) \otimes \Id_{\cQ_\omega^d}\\
		&= \sum_{\mu\prec\lambda} \sum_{i,j}(\tau_\mu^\lambda)_{ij} \bigoplus_{\omega\vdash_d m-1} \sum_{\pi\in \SC_{m-1}}p_\mu(\pi)_{ji} p_\omega(\pi)_{kl} |k\rangle \langle l|_{\cP_\omega} \otimes \Id_{\cQ_\omega^d}\\
		&= \sum_{\mu\prec\lambda} \frac{(m-1)!}{f_\mu} \sum_{i,j} (\tau_\mu^\lambda)_{ij}  |j\rangle\langle i| \otimes \Id_{Q_\mu^d}\label{eq:schur-orth}\\
		&= \sum_{\mu\prec\lambda} \frac{(m-1)!}{f_\mu} (\tau_\mu^\lambda)^T \otimes \Id_{Q_\mu^d}\\
		&= \sum_{\mu\prec\lambda} \frac{(m-1)!}{f_\mu} \tau_\mu^\lambda\otimes \Id_{Q_\mu^d} \label{eq:mu-component-reduction}
	\end{align}
	where we used Schur's orthogonality relation $\sum_{\pi\in \SC_{m-1}}p_\mu(\pi)_{ji} p_\omega(\pi)_{kl} = \frac{(m-1)!}{f_\mu} \delta_{\mu\omega}\delta_{jk}\delta_{il}$ in \eqref{eq:schur-orth}.
	
	To compute $\tau_\mu^\lambda$, we follow \cite{vershik2005new} and consider the Gelfand-Tsetlin basis $\lbrace |T\rangle : T\in\SYT(\lambda)\rbrace$ for the $S_m$-irrep $\cP_\lambda$, where $\SYT(\lambda)$ is the set of standard Young tableaux (SYT) of shape $\lambda$.
	In any $T\in \SYT(\lambda)$ the label \texttt{m} can only appear in a box to the very right of a column and such that there is no box below it; in other words, the possible locations of \texttt{m} are the boxes that can be removed from $\lambda$ to obtain another Young diagram with $m-1$ boxes.
	Fixing one of these positions for \texttt{m}, the subset of SYTs with \texttt{m} in that position is an orthonormal basis for an isomorphic copy of the irreducible representation $\cP_{\mu}$ inside $\cP_\lambda$, where $\mu$ is obtained from $\lambda$ by removing the box with the label \texttt{m} in $T$.
	Likewise, the subset of SYTS with the labels \texttt{m-1} and \texttt{m} in fixed positions is an orthonormal basis for an isomorphic copy of $\cP_\nu$ inside $\cP_\mu$, where now $\nu$ is obtained from $\mu$ by removing the box with the label \texttt{m-1} in $T$.
	Denoting by $T(\lambda\setminus\mu)$ the label of the box in $\lambda\setminus\mu$, the map $P_\mu^\lambda$ projecting down onto a particular copy of $\cP_\mu$ in $\cP_\lambda$ is therefore given by 
	\begin{align}
		P_\mu^\lambda = \sum_{\substack{T\in\SYT(\lambda)\colon\\ T(\lambda\setminus\mu)=\texttt{m}}} |T'\rangle\langle T|,
        \label{eq:project-down}
	\end{align}
	where we denote by $T'$ the standard Young tableau of shape $\mu$ obtained from $T$ be removing the box labeled with \texttt{m}.
	We can then write
	\begin{align}
		\tau_\mu^\lambda = P_\mu^\lambda p_\lambda(\tau) (P_\mu^\lambda)^\dagger = \sum_{\substack{S\in\SYT(\lambda)\colon\\ S(\lambda\setminus\mu)=\texttt{m}}}\sum_{\substack{T\in\SYT(\lambda)\colon\\ T(\lambda\setminus\mu)=\texttt{m}}} |S'\rangle\langle S| p_\lambda(\tau) |T\rangle\langle T'|
		\label{eq:projected-transposition}
	\end{align}
	
	According to \cite[Prop.~6.2]{vershik2005new}, the transposition $\tau=(m-1,m)$ acts on a basis vector $|T\rangle$ of $\cP_\lambda$ in two possible ways: 
	\begin{align}
		p_\lambda(\tau)|T\rangle = \begin{cases}
			r^{-1}|T\rangle & \text{if \texttt{m-1} and \texttt{m} are in the same row or column;}\\
			r^{-1} |T\rangle + \sqrt{1-r^{-2}} |\tau T\rangle & \text{otherwise,}
		\end{cases}
	\end{align}
	where $r=c_T(\texttt{m})-c_T(\texttt{m-1})$, with $c_T(\texttt{j})$ denoting the content of the box holding the label~\texttt{j}, and $\tau T$ is the SYT obtained from $T$ by switching \texttt{m-1} and \texttt{m}.
	Note that we have $\langle S| \tau T\rangle=0$ for any SYT $S$ with \texttt{m} in the same position as in $T$.
	Hence, \eqref{eq:projected-transposition} simplifies to
	\begin{align}
		\tau_\mu^\lambda &= \sum_{\substack{S\in\SYT(\lambda)\colon\\ S(\lambda\setminus\mu)=\texttt{m}}}\sum_{\substack{T\in\SYT(\lambda)\colon\\ T(\lambda\setminus\mu)=\texttt{m}}} |S'\rangle\langle S| p_\lambda(\tau) |T\rangle\langle T'|\\
		&= \sum_{\substack{S\in\SYT(\lambda)\colon\\ S(\lambda\setminus\mu)=\texttt{m}}}\sum_{\substack{T\in\SYT(\lambda)\colon\\ T(\lambda\setminus\mu)=\texttt{m}}} \frac{1}{c(\lambda\setminus\mu) - c_T(\texttt{m-1})} |S'\rangle\langle S | T\rangle\langle T'|\\
		&= \sum_{\substack{T\in\SYT(\lambda)\colon\\ T(\lambda\setminus\mu)=\texttt{m}}} \frac{1}{c(\lambda\setminus\mu) - c_T(\texttt{m-1})} | T'\rangle\langle T'|\\
		&= \sum_{\nu\prec\mu} \frac{1}{c(\lambda\setminus\mu) - c(\mu\setminus\nu))} \sum_{\substack{T\in\SYT(\lambda)\colon\\ T(\lambda\setminus\mu)=\texttt{m}\text{ and }T(\mu\setminus\nu)=\texttt{m-1}}} |T'\rangle \langle T'|\\
		&= \sum_{\nu\prec\mu} \frac{1}{c(\lambda\setminus\mu) - c(\mu\setminus\nu)} P_\nu^\mu,
	\end{align}
    where $P_\nu^\mu$ is defined in analogy to \eqref{eq:project-down}.
	
	Substituting this expression for $\tau_\mu^\lambda$ in \eqref{eq:mu-component-reduction} gives
	\begin{align}
		\sum_{\pi\in G} \chi_\lambda(\tau\pi) P_d'(\pi) &= \sum_{\mu\prec\lambda} \frac{(m-1)!}{f_\mu} \tau_\mu^\lambda\otimes \Id_{Q_\mu^d}\\
		&= \sum_{\mu\prec\lambda} \frac{(m-1)!}{f_\mu} \sum_{\nu\prec\mu} \frac{1}{c(\lambda\setminus\mu) - c(\mu\setminus\nu)} P_\nu^\mu \otimes \Id_{Q_\mu^d}\\
		&= \sum_{\mu\prec\lambda} \frac{(m-1)!}{f_\mu} \sum_{\nu\prec\mu} \frac{1}{c(\lambda\setminus\mu) - c(\mu\setminus\nu)} \left(\Pi^{(m-2)}_\nu\otimes \Id_{m-1}\right)\Pi^{(m-1)}_\mu,
	\end{align}
	which proves the claim.
\end{proof}

\subsection{Counterexample to Tightening Harrow's Coefficient Bound}
Harrow's coefficient estimate cannot generally be strengthened from $O(d^{-\lvert\pi\rvert/2})$ to $O(d^{-\lvert\pi\rvert})$. To see this, let $\tau\coloneq(123)\in\SC_3$ and, for $d\geq3$, define
\begin{align}
    M
    \coloneq
    \frac{1}{2}\Id
    +
    \frac{1}{4d}
    \left(
        P_d(\tau)+P_d(\tau^{-1})
    \right).
\end{align}
For every $S\subseteq\{1,2,3\}$, a three-cycle crosses the cut $S:S^c$ at most once in either direction. The singular-value formula for partially transposed permutation operators from Ref.~\cite[Lem.~4]{harrow2023Approximate} therefore gives
\begin{align}
    \norm{P_d(\tau)^{\Gamma_S}}_\infty
    \leq d,
    \qquad
    \norm{P_d(\tau^{-1})^{\Gamma_S}}_\infty
    \leq d.
\end{align}
Consequently,
\begin{align}
    \norm{
        M^{\Gamma_S}-\frac{1}{2}\Id
    }_\infty
    &\leq
    \frac{1}{4d}
    \left(
        \norm{P_d(\tau)^{\Gamma_S}}_\infty
        +
        \norm{P_d(\tau^{-1})^{\Gamma_S}}_\infty
    \right) \leq
    \frac{1}{2}.
\end{align}
Since partial transposition preserves Hermiticity, this implies
\begin{align}
    0\preceq M^{\Gamma_S}\preceq\Id
    \quad
    \text{for every} \quad S\subseteq\{1,2,3\}.
\end{align}
Thus, $\{M,\Id-M\}$ is PPT-BOTH. In the stable range $d\geq3$, the permutation operators are linearly
independent, so the above expansion is unique and
\begin{align}
    m_\tau=m_{\tau^{-1}}=\frac{1}{4d}.
\end{align}
However, $\lvert\tau\rvert=2$. Hence, a uniform coefficient estimate
of the form
\begin{align}
    \abs{m_\pi}
    \leq
    C_n d^{-\lvert\pi\rvert},
\end{align}
with $C_n$ independent of $d$, would require
\begin{align}
    \frac{1}{4d}
    \leq
    \frac{C_3}{d^2},
\end{align}
which fails for arbitrarily large $d$. Therefore, the
$O(n^2/d)$ data-hiding bound cannot be obtained merely by improving
Harrow's individual coefficient estimates to
$d^{-\lvert\pi\rvert}$; instead, one must control the distinguishing
functional more collectively.

\end{document}